\documentclass{article}
\usepackage[utf8]{inputenc}
 \usepackage[T1]{fontenc}
\usepackage{bm}
\usepackage{times}
\usepackage{caption}
\usepackage{enumerate}
\usepackage{amsmath,amssymb}
\usepackage{color}
\usepackage[table]{xcolor}
\usepackage{relsize}
\usepackage{graphicx}
\usepackage{subfigure}
\usepackage{mathtools}
\usepackage{braket}
\usepackage{hyperref}
\usepackage[qm]{qcircuit}
\usepackage{array}
\usepackage{algorithm,algorithmic}
\usepackage{bm}
\usepackage[titletoc]{appendix}
\hypersetup{colorlinks=true,citecolor=blue,linkcolor=blue,filecolor=blue,urlcolor=blue,breaklinks=true}

\newtheorem{definition}{Definition}[section]
\newtheorem{proposition}[definition]{Proposition}
\newtheorem{lemma}[definition]{Lemma}

\newtheorem{theorem}[definition]{Theorem}
\newtheorem{corollary}[definition]{Corollary}

\newtheorem{example}[definition]{Example}

\newenvironment{proof}{\noindent \textbf{{Proof.~} }}{\hfill $\blacksquare$}
\definecolor{colorone}{rgb}{1,0.36,0.03}
\definecolor{colortwo}{rgb}{0.54,0.71,0.03}
\definecolor{colorthree}{rgb}{0.01,0.51,0.93}
\definecolor{colorfour}{rgb}{0.47,0.26,0.58}

\newcommand{\tr}{\operatorname{Tr}}
\newcommand{\proj}[1]{| #1\rangle\!\langle #1 |}

\newcommand{\ie}{\textit{i}.\textit{e}.}

\usepackage{geometry} 
\title{All real projective measurements can be self-tested}
\author{Ranyiliu Chen \and Laura Man\v{c}inska \and Jurij Vol\v{c}i\v{c}}
\date{}

\begin{document}
\maketitle

\begin{abstract} 
Self-testing is the strongest form of quantum functionality verification which allows a classical user to deduce the quantum state and measurements used to produce measurement statistics.
While self-testing of quantum states is well-understood, self-testing of measurements, especially in high dimensions, has remained more elusive. We demonstrate the first general result in this direction by showing that every real projective measurement can be self-tested. The standard definition of self-testing only allows for the certification of real measurements. Therefore, our work effectively broadens the scope of self-testable projective measurements to their full potential. To reach this result, we employ the idea that existing self-tests can be extended to verify additional untrusted measurements. This is known as `post-hoc self-testing'. 
We formalize the method of post-hoc self-testing and establish a sufficient condition for its application. Using this condition we construct self-tests for all real projective measurements. Inspired by our construction, we develop a new technique of \emph{iterative self-testing}, which involves using post-hoc self-testing in a sequential manner. Starting from any established self-test, we fully characterize the set of measurements that can be verified via iterative self-testing. This provides a clear methodology for constructing new self-tests from pre-existing ones.
\end{abstract}

\section{Introduction}
\label{sec:intro}

\subsection{Background}

Consider a scenario where a classical user, Victor, engages with a quantum device by posing questions $x$ and receiving answers $a$. Lacking any prior knowledge of the device's internal workings, Victor models its behavior as a state preparation $\ket{\psi}$, accompanied by quantum measurements $\{M_{a|x},\sum_aM_{a|x}=I\}$. In response to question $x$, the device executes measurement $\{M_{a|x}\}_a$ on the state $\ket{\psi}$ and outputs the resulting measurement output $a$. While it is straightforward to predict the device's output statistics from $\ket{\psi}$ and $\{M_{a|x}\}$ using Born's rule \cite{born_zur_1926} $p(a|x)=\braket{\psi|M_{a|x}|\psi}$, it is impossible to deduce $\ket{\psi}$ and $\{M_{a|x}\}$ solely from the statistics $p(a|x)$. Indeed, different states $\ket{\psi}$ and $\{M_{a|x}\}$ can yield the same $p(a|x)$. In this setting, even a classical computer is always able to simulate the quantum process, if its running time is not limited.

Intriguingly, deducing the quantum functionality from the resulting classical statistics becomes possible in the so-called `bipartite Bell scenario' {\cite{Brunner_2014,Bell1964paradox}} (see Fig. \ref{fig:1}). Here, Victor interacts with \emph{two} spatially separated quantum devices, named Alice and Bob. He poses questions $x\in\mathcal{I}_A$ and $y\in\mathcal{I}_B$ to Alice and Bob respectively, who in turn provide answers, $a\in \mathcal{O}_A$ and $b\in \mathcal{O}_B$. While Alice and Bob cannot communicate during this interaction, they may share an entangled quantum state $\ket{\psi}_{AB}$, which they can measure locally using measurements $\{M_{a|x}:a\in \mathcal{O}_A,x\in \mathcal{I}_A\}$ and $\{N_{b|y}:b\in \mathcal{O}_B,y\in \mathcal{I}_B\}$ to obtain outputs $a$ and $b$. The statistics observed by Victor then follow the distribution $p(a,b|x,y)=\braket{\psi|M_{a|x}\otimes N_{b|y}|\psi}$. It turns out that some statistics $p(a,b|x,y)$ can \emph{exclusively} be produced by a specific set of measurements $\{M_{a|x}\}$ and $\{N_{b|y}\}$ on a specific entangled state $\ket{\psi}_{AB}$ (up to a change of a local frame of reference). This phenomenon is known as `self-testing' \cite{Mayers04self} and it relies on key features of the quantum theory such as entanglement \cite{GISIN199215} and incompatibility of measurements \cite{PhysRevLett.103.230402}. Self-testing represents the strongest form of verification as it requires minimal assumptions, namely, no-communication between Alice's and Bob's measuring devices and the validity of the quantum theory. In particular, in self-testing we do not require access to any trusted or fully characterized quantum devices, a condition also known as `device-independence' \cite{PhysRevLett.98.230501}.

\begin{figure}[ht]%
\centering
\includegraphics[width=0.95\textwidth]{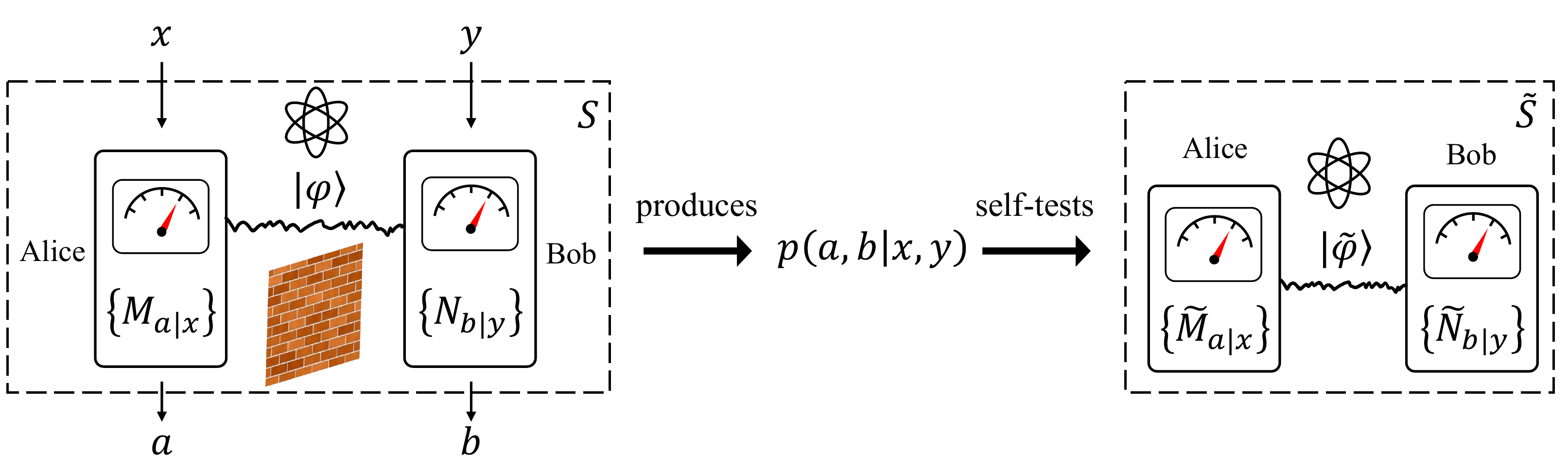}
\caption{Self-testing in a Bell scenario. Spatially separated Alice and Bob perform local measurements on a shared state (described by $\mathcal{S}$), giving rise to correlation $p(a,b|x,y)$. In the case of self-testing, Victor can classically verify Alice and Bob: the only way for Alice and Bob to produce the correct correlation is by adhering to the prescribed specification $\tilde{\mathcal{S}}$.}\label{fig:1}
\end{figure}


The quantum mechanical description of the devices in a bipartite Bell scenario is given by what we call a \emph{strategy}. Formally, such a strategy  ${\mathcal{S}}$ is a tuple
\begin{align*}
    \mathcal{S}=(\ket{\psi}_{AB},\{M_{a|x}\colon a\in \mathcal{O}_A,x\in \mathcal{I}_A\},\{N_{b|y}\colon b\in \mathcal{O}_B,y\in \mathcal{I}_B\}),
\end{align*}
where $\ket{\psi}_{AB}\in\mathcal{H}_A\otimes\mathcal{H}_B$ is the shared state, $\mathcal{M}_x=\{M_{a|x}\}\subseteq\mathcal{L}(\mathcal{H}_A)$ and $\mathcal{N}_b=\{N_{b|y}\}\subseteq\mathcal{L}(\mathcal{H}_B)$ are the POVMs (positive operator-valued measures) of Alice and Bob respectively. The resulting measurement statistics
\begin{align*}
    p(a,b|x,y)=\braket{\psi|M_{a|x}\otimes N_{b|y}|\psi}
\end{align*}
is commonly referred to as \emph{correlation}.
In case of self-testing, we can recover the description of the state and measurements comprising $\mathcal{S}$ from merely observing the measurement statistics $p$ that it produces. So whenever self-testing holds, we can verify the involved state-preparation and measurement functionalities without any prior knowledge about the inner workings of the employed quantum devices.
This leads us to the following fundamental question of self-testing:\vspace{.05in}

\noindent{\bf Question:} \emph{Which quantum states and which measurements can be self-tested?}\vspace{.02in}

In other words, the above question asks which state-preparation and measurement functionalities can be verified by a classical user with no access to trusted quantum devices. To verify (self-test) a given state-preparation or measurement functionality, we need to construct a strategy $\mathcal{S}$ which incorporates this functionality and is moreover determined (self-tested) by the correlation it produces.

In the bipartite scenario the question regarding self-testable states has been answered by a  milestone result~\cite{Coladangelo2017all}, which allows to self-test any pure bipartite entangled state. Recent work shows that in the network setting {\cite{Tavakoli_2022}} it is possible to self-test any entangled multiparty state~\cite{supic2023}. In contrast to this rather complete picture for self-testing of quantum states, the self-testing of general measurements has remained elusive. First, in the bipartite Bell scenario, any complex measurement, \emph{i.e.}, having complex matrix entries in the Schmidt basis of the shared state, can be simulated by a real measurement \cite{P_l_2008,McKague_2009,renou_quantum_2021}. So any self-tested measurement must be real, because a measurement always produces the same statistics as its complex conjugate:
$$\braket{\psi|M_{a|x}\otimes N_{b|y}|\psi}=\overline{\braket{\psi|M_{a|x}\otimes N_{b|y}|\psi}}=\braket{\psi|\overline{M}_{a|x}\otimes \overline{N}_{b|y}|\psi}.$$ 
Existing protocols primarily focus on low-dimensional quantum systems or specific higher-dimensional measurements. In case of a two-level system, it is known how to self-test Pauli measurements~\cite{Mayers04self}, and subsequent work has shown that any two-dimensional projective measurement is self-testable \cite{Yang_2013}. In \cite{McKague2017selftestingin,Andrea2017parallel} tensor-products of Pauli matrices are self-tested, and \cite{sarkar2021self} presented a self-test for a particular pair of $d$-output measurements. In \cite{Laura2021Constant,Fu_2022}, constant-sized self-test of measurements satisfying some special property is provided. Self-testing of arbitrary higher-dimensional measurements, however, has remained out of reach up to now.

\subsection{Results}

In this work we study the self-testing of measurements in a comprehensive (as opposed to example-based) manner. We provide the first results for self-testing of general measurements.
Our specific contributions include the following:

\begin{enumerate}[(i)]
    \item We put forth a fully explicit self-testing protocol for any real projective measurement. Our construction has a question set of cubic size in $d$, the dimension of the measurement to be self-tested, and a constant-sized answer set. Additionally, our self-test is robust to noise.
    \item We formalize the method of \emph{post-hoc self-testing} and identify a sufficient condition for its application. Post-hoc self-testing occurs when we can extend a previously self-tested strategy to include an additional measurement. While
    there are sporadic examples of this method in the literature, a comprehensive understanding of this phenomenon and when it occurs was lacking. To remedy this, we have identified a sufficient condition under which an initial self-test of a given strategy $\mathcal{S}$ can be extended to include an additional measurement $\mathcal{M}$.
    Applying this criterion to an initial strategy from the recent work \cite{Laura2021Constant} allows us to obtain our self-testing construction from (i).
    \item We develop a new technique of \emph{iterative self-testing} which involves sequential application of post-hoc self-testing. Starting from any established self-test, we use Jordan algebra to characterize the set of measurements that can be verified via iterative self-testing. Iterative self-testing is inspired by our self-testing construction in (ii), and offers a handy way for developing new self-tests based on pre-existing ones.

\end{enumerate}

\subsection{Organisation of the paper}

We give the preliminary for this paper and the notations we frequently used herein in Section \ref{sec:pre}. In section \ref{sec:post} we formalize the definition of post-hoc self-testing, and provide a criterion to determine whether a projective measurement is post-hoc self-tested. Then in Section \ref{sec:d+1} we consider iterative self-testing, by which we show how to self-test any real projective measurements. In Section \ref{sec:iterative} the general theory of iterative self-testing is presented. We conclude our result and provide some open problems in Section \ref{sec:conclude}. Appendix \ref{sec:appA} contains examples, counter-examples and general obstructions of post-hoc self-testing. Appendix \ref{sec:appB} is a construction of the observables in our robust self-tested strategy with Mathematica code.

\section{Preliminaries and notation}
\label{sec:pre}

\subsection{Oobservable picture of measurements}

In many cases, especially when the measurement is projective (\ie, all operators in the POVMs are projections), it can be more convenient to work with \emph{generalized observables} instead of operators of POVMs. Given a POVM $\{M_{a|x}\}$, its generalized observables are given by
$$
A_x^{(j)}:=\sum_{a=0}^{|\mathcal{O}_A|-1}\omega^{aj}M_{a|x},
$$
where $\omega=e^{i2\pi/|\mathcal{O}_A|}$. Note that $A_x^{(0)}=I$ by definition. Due to the invertibility of the transform, $\{M_{a|x}\}$ can be recovered from $\{A_x^{(j)}\}$ by $M_a=\frac{1}{|\mathcal{O}_A|}\sum_{j=0}^{|\mathcal{O}_A|-1}\omega^{-aj}A_x^{(j)}$. So $\{A_x^{(j)}\}$ provides an alternative, yet full, description of the measurement. The following properties about the generalized observables hold: (see \cite{Kaniewski2019maximalnonlocality} for a proof)
\begin{itemize}
    \item For any POVM $\{M_{a|x}\}$, $A_x^{(j)}A_x^{(j)\dagger}\le I, A_x^{(j)\dagger} A_x^{(j)}\le I$, \ie, $A_x^{(j)}$ are contractions.
    \item A POVM $\{M_{a|x}\}$ is projective if and only if the corresponding $A_x:=A_x^{(1)}$ is a unitary matrix of order $|\mathcal{O}_A|$. In this case, we call $A_x$ the \emph{observable} of $\{M_{a|x}\}$, further having that $A_x^{(j)}=A_x^{j}$. Therefore,
    \item \emph{Projective} measurements are fully characterised by its observable: $$M_a=\frac{1}{|\mathcal{O}_A|}\sum_{j=0}^{|\mathcal{O}_A|-1}\omega^{-aj}A_x^{j},$$
    while in general, it might not be possible to recover the POVM elements of a measurement from $A_x$.
\end{itemize}

In this work, we specify quantum strategies by the tuple 
\begin{align*}
    \mathcal{S}=(\ket{\psi}_{AB},\{A_{x}^{(j)}:x\in \mathcal{I}_A,j\in\mathcal{O}_A\},\{B_{y}^{(k)}:y\in \mathcal{I}_B,k\in\mathcal{O}_B\}),
\end{align*}
where $A_x^{(j)}=\sum_{a=0}^{|\mathcal{O}_A|-1}\omega_A^{aj}M_{a|x}$, $\omega_A=e^{i2\pi/|\mathcal{O}_A|}$, $B_y^{(k)}=\sum_{b=0}^{|\mathcal{O}_B|-1}\omega_B^{bk}N_{b|y}$, $\omega_B=e^{i2\pi/|\mathcal{O}_B|}$. The correlation is also conveniently specified via
$$\{\braket{\psi|A^{(j)}_x\otimes B^{(k)}_y|\psi}\}_{j,k,x,y}=\{\sum_{a,b}\omega_A^{aj}\omega_B^{bk}p(ab|xy)\}_{j,k,x,y}.$$

Furthermore, we call $\mathcal{S}$ projective if all the measurements in $\mathcal{S}$ are all projective, and denote it by $\mathcal{S}=(\ket{\psi}_{AB},\{A_{x}:x\in \mathcal{I}_A\},\{B_{y}:y\in \mathcal{I}_B\})$ for simplicity. In this work we shall present our results in terms of observables.

\subsection{Local dilation and self-testing}

In a self-testing protocol the verifier Victor wishes to infer the underlying quantum strategy from his observation of correlations, so it is desired that the strategy generating a given correlation is to some extent unique. However, there are at least two types of manipulation of the strategy that do not affect the correlation. Firstly, if we only choose a different basis, then strategies $\mathcal{S}=(\ket{\psi}_{AB},\{A_x^{(j)}\},\{B_y^{(k)}\})$ and $\mathcal{S}'=(U_A\otimes U_B\ket{\psi}_{AB},\{U_AA_x^{(j)}U_A^\dagger\},\{U_BB_y^{(k)}U_B^\dagger\})$ produce the same correlation for any local unitaries $U_A,U_B$. Secondly, if we attach a bipartite auxiliary state $\ket{\operatorname{aux}}_{A'B'}$ on which the measurements act trivially, then strategies $\mathcal{S}=(\ket{\psi}_{AB},\{A_x^{(j)}\},\{B_y^{(k)}\})$ and $\mathcal{S}'=(\ket{\operatorname{aux}}_{A'B'}\otimes\ket{\psi}_{AB},\{I\otimes A_x^{(j)}\},\{I\otimes B_y^{(k)}\})$ produce the same correlation. Motivated by the above two manipulations, we say that `$\tilde{\mathcal{S}}$ is a local dilation of ${\mathcal{S}}$' if up to a change of local bases $\mathcal{S}$ is $\tilde{\mathcal{S}}$ plus some trivial auxiliary state. Specifically, we say that 

\begin{definition}[local dilation]
A local operator $\tilde{M}_{\tilde{A}}\otimes \tilde{N}_{\tilde{B}}$ and a bipartite state $\ket{\tilde{\psi}}_{\tilde{A}\tilde{B}}$ is a \emph{local dilation} of operator $M_A\otimes N_B$ and state $\ket{\psi}_{AB}$ with an auxiliary state $\ket{\operatorname{aux}}_{A'B'}\in\mathcal{H}_{A'}\otimes\mathcal{H}_{B'}$ and a local isometry $\Phi=\Phi_A\otimes\Phi_B$ where $\Phi_A:\mathcal{H}_A\rightarrow\mathcal{H}_{{A'}}\otimes\mathcal{H}_{\tilde{A}}$, $\Phi_B:\mathcal{H}_B\rightarrow\mathcal{H}_{{B'}}\otimes\mathcal{H}_{\tilde{B}}$, if 
\begin{align*}
    &\Phi[M_{A}\otimes N_B\ket{\psi}_{AB}]=\ket{\operatorname{aux}}_{A'B'}\otimes(\tilde{M}_{\tilde{A}}\otimes \tilde{N}_{\tilde{B}}\ket{\tilde{\psi}}_{\tilde{A}\tilde{B}}).
\end{align*}
We denote local dilation by
\begin{align*}
    (\ket{\psi}_{AB},M_A\otimes N_B)\xhookrightarrow{\Phi,\ket{\operatorname{aux}}}(\ket{\tilde{\psi}}_{\tilde{A}\tilde{B}},\tilde{M}_{\tilde{A}}\otimes \tilde{N}_{\tilde{B}}),
\end{align*}
omitting ${\Phi},\ket{\operatorname{aux}}$ above the arrow if they are clear from the context.

A strategy $\tilde{\mathcal{S}}=(\ket{\tilde{\psi}}_{\tilde{A}\tilde{B}},\{\tilde{A}_{x}^{(j)}),\{\tilde{B}_{y}^{(k)}\})$ is a \emph{local dilation} of strategy ${\mathcal{S}}=(\ket{\psi}_{AB},\{A_x^{(j)}),\{B_y^{(k)}\})$ with an auxiliary state $\ket{\operatorname{aux}}_{A'B'}$ and a local isometry $\Phi=\Phi_A\otimes\Phi_B$ where $\Phi_A:\mathcal{H}_A\rightarrow\mathcal{H}_{{A'}}\otimes\mathcal{H}_{\tilde{A}}$, $\Phi_B:\mathcal{H}_B\rightarrow\mathcal{H}_{{B'}}\otimes\mathcal{H}_{\tilde{B}}$, if
\begin{align*}
    &(\ket{\psi}_{AB},A^{(j)}_{x}\otimes I_B)\xhookrightarrow{\Phi,\ket{\operatorname{aux}}}(\ket{\tilde{\psi}}_{\tilde{A}\tilde{B}},\tilde{A}^{(j)}_{x}\otimes I_{\tilde{B}}),
    \\
    &(\ket{\psi}_{AB},I_A\otimes B^{(k)}_{y})\xhookrightarrow{\Phi,\ket{\operatorname{aux}}}(\ket{\tilde{\psi}}_{\tilde{A}\tilde{B}},I_{\tilde{A}}\otimes\tilde{B}^{(k)}_{y})
\end{align*}
hold for all $j\in\mathcal{O}_A,k\in\mathcal{O}_B,x\in\mathcal{I}_A,y\in\mathcal{I}_B$. We denote local dilation of strategies by $\mathcal{S}\xhookrightarrow{\Phi,\ket{\operatorname{aux}}} \tilde{\mathcal{S}}$.
\label{def:local_dilationobs}
\end{definition}

We also use the hook arrow notation for approximate local dilation:

\begin{definition}[local $\varepsilon$-dilation]
A local operator $\tilde{M}_{\tilde{A}}\otimes \tilde{N}_{\tilde{B}}$ and a bipartite state $\ket{\tilde{\psi}}_{\tilde{A}\tilde{B}}$ is a \emph{local $\varepsilon$-dilation} of operator $M_A\otimes N_B$ and state $\ket{\psi}_{AB}$ with an auxiliary state $\ket{\operatorname{aux}}_{A'B'}\in\mathcal{H}_{A'}\otimes\mathcal{H}_{B'}$ and a local isometry $\Phi=\Phi_A\otimes\Phi_B$ where $\Phi_A:\mathcal{H}_A\rightarrow\mathcal{H}_{{A'}}\otimes\mathcal{H}_{\tilde{A}}$, $\Phi_B:\mathcal{H}_B\rightarrow\mathcal{H}_{{B'}}\otimes\mathcal{H}_{\tilde{B}}$, if 
\begin{align*}
    &\|\Phi[M_{A}\otimes N_B\ket{\psi}_{AB}]-\ket{\operatorname{aux}}_{A'B'}\otimes(\tilde{M}_{\tilde{A}}\otimes \tilde{N}_{\tilde{B}}\ket{\tilde{\psi}}_{\tilde{A}\tilde{B}})\|\le\varepsilon.
\end{align*}
We denote local $\varepsilon$-dilation by
\begin{align*}
    (\ket{\psi}_{AB},M_A\otimes N_B)\xhookrightarrow[\varepsilon]{\Phi,\ket{\operatorname{aux}}}(\ket{\tilde{\psi}}_{\tilde{A}\tilde{B}},\tilde{M}_{\tilde{A}}\otimes \tilde{N}_{\tilde{B}}).
\end{align*}
A strategy $\tilde{\mathcal{S}}=(\ket{\tilde{\psi}}_{\tilde{A}\tilde{B}},\{\tilde{A}^{(j)}_{x}),\{\tilde{B}^{(k)}_{y}\})$ is a \emph{local $\varepsilon$-dilation} of a strategy ${\mathcal{S}}=(\ket{\psi}_{AB},\{A^{(j)}_x),\{B^{(k)}_y\})$ with an auxiliary state $\ket{\operatorname{aux}}_{A'B'}$ and a local isometry $\Phi=\Phi_A\otimes\Phi_B$ where $\Phi_A:\mathcal{H}_A\rightarrow\mathcal{H}_{{A'}}\otimes\mathcal{H}_{\tilde{A}}$, $\Phi_B:\mathcal{H}_B\rightarrow\mathcal{H}_{{B'}}\otimes\mathcal{H}_{\tilde{B}}$, if
\begin{align*}
    &(\ket{\psi}_{AB},A^{(j)}_{x}\otimes I_B)\xhookrightarrow[\varepsilon]{\Phi,\ket{\operatorname{aux}}}(\ket{\tilde{\psi}}_{\tilde{A}\tilde{B}},\tilde{A}^{(j)}_{x}\otimes I_{\tilde{B}}),
    \\
    &(\ket{\psi}_{AB},I_A\otimes B^{(k)}_{y})\xhookrightarrow[\varepsilon]{\Phi,\ket{\operatorname{aux}}}(\ket{\tilde{\psi}}_{\tilde{A}\tilde{B}},I_{\tilde{A}}\otimes\tilde{B}^{(k)}_{y})
\end{align*}
hold for all $j\in\mathcal{O}_A,k\in\mathcal{O}_B,x\in\mathcal{I}_A,y\in\mathcal{I}_B$. We denote local $\varepsilon$-dilation of strategies by $\mathcal{S}\xhookrightarrow[\varepsilon]{\Phi,\ket{\operatorname{aux}}} \tilde{\mathcal{S}}$.
\label{def:local_dilationobsapp}
\end{definition}

Now we can define exact and robust self-testing of strategies $\tilde{\mathcal{S}}$:

\begin{definition}[self-testing of strategy]
A strategy $\tilde{\mathcal{S}}=(\ket{\tilde{\psi}}_{\tilde{A}\tilde{B}},\{\tilde{A}^{(j)}_{x}),\{\tilde{B}^{(k)}_{y}\})$ is \emph{(exactly) self-tested} by the correlation $\{\braket{\tilde{\psi}|\tilde{A}^{(j)}_x\otimes \tilde{B}^{(k)}_y|\tilde{\psi}}\}$ it generates, if any strategy $\mathcal{S}$ generating the same correlation $\{\braket{\tilde{\psi}|\tilde{A}^{(j)}_x\otimes \tilde{B}^{(k)}_y|\tilde{\psi}}\}$ satisfies $\mathcal{S}\xhookrightarrow{\Phi,\ket{\operatorname{aux}}}\tilde{\mathcal{S}}$ for some local isometry $\Phi$ and auxiliary state $\ket{\operatorname{aux}}$.

A strategy $\tilde{\mathcal{S}}=(\ket{\tilde{\psi}}_{\tilde{A}\tilde{B}},\{\tilde{A}^{(j)}_{x}),\{\tilde{B}^{(k)}_{y}\})$ is \emph{robustly self-tested} by the correlation it generates, if it is self-tested and the following condition holds: for any $\varepsilon>0$, there exists $\delta>0$ such that, any strategy $\mathcal{S}$ $\delta$-approximately generating the correlation (\ie, for all $x,y,j,k$ it holds that $|\braket{\psi|A_x^{(j)}\otimes B^{(k)}_y|\psi}- \braket{\tilde{\psi}|\tilde{A}^{(j)}_x\otimes \tilde{B}^{(k)}_y|\tilde{\psi}}|<\delta$) satisfies
$\mathcal{S}\xhookrightarrow[\varepsilon]{\Phi,\ket{\operatorname{aux}}}\tilde{\mathcal{S}}$ for some local isometry $\Phi$ and auxiliary state $\ket{\operatorname{aux}}$. 
\label{def:realstrob}
\end{definition}

In the context of self-testing, $\tilde{\mathcal{S}}$ is often called the reference strategy, and $\mathcal{S}$ is called the physical one. 

We make the following assumptions in this work. 
Firstly, we consider `real' reference strategies, due to the complex-conjugation issue of self-testing measurement (see Section 3.7.1 of \cite{Supic2020selftestingof} for a detailed discussion). By `real' we mean that in a Schmidt basis of $\ket{\tilde{\psi}}$ the matrices of the POVMs in $\tilde{\mathcal{S}}$ have real entries. Secondly, we consider projective reference strategies. Recall that they are fully characterised by their observables. In this sense, by `self-testing of a observable' we essentially mean the `self-testing of the projective measurement it corresponds to'. We note that, while some previous self-testing results also make assumption of projective physical strategies, we do not require the physical strategy to be projective or real. Finally, we without loss of generality assume that left and right local dimension of the `auxiliary' state $\ket{\operatorname{aux}}$ are the same, \ie, $\mathcal{H}_{A'}\cong\mathcal{H}_{B'}$, since one can always enlarge the smaller marginal space with local isometry. We also consider full-Schmidt rank reference state $\ket{\tilde{\psi}}$; note that this is a very common assumption throughout the literature on self-testing. Therefore, we can take $\mathcal{H}_{\tilde{A}}\cong\mathcal{H}_{\tilde{B}}$.

\section{Robust post-hoc self-testing of projective measurements}
\label{sec:post}

The concept of post-hoc self-testing has been implicitly employed in prior works, such as self-testing of graph states \cite{McKague2016interactive}, randomness certification \cite{Andersson2018Device,Woodhead2020maximal}, and one-sided self-testing \cite{Sarkar_2023}. The review paper \cite{Supic2020selftestingof} was the first to summarize this technique and refer to it as `post-hoc self-testing'. In this section, we formalise the idea of post-hoc self-testing and establish a sufficient condition for its application.

\subsection{Definitions}

\begin{figure*}[t]
	\centering
	\includegraphics[width=1.0\textwidth]{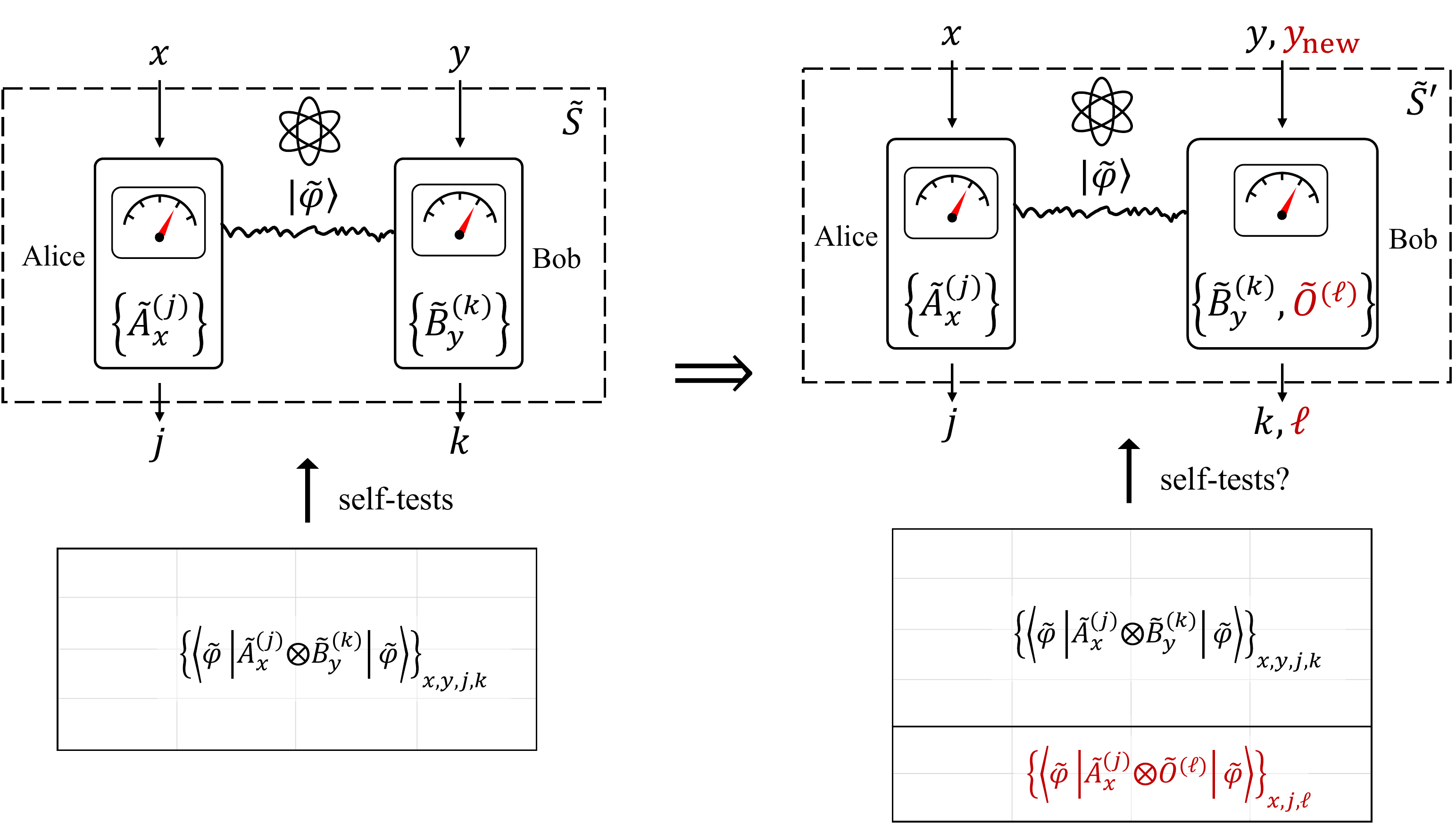}
	\caption{Post-hoc self-testing: starting from a self-tested strategy $\tilde{\mathcal{S}}$ (on the left), if it is feasible to infer the new measurement $\tilde{O}^{(\ell)}$ with input $y_{\text{new}}$ and output $\ell$ from correlations $\{\braket{\psi|A_x^{(j)}\otimes O^{(\ell)}|\psi}\}$, then extended strategy $\tilde{\mathcal{S}}'$ (on the right) remains self-tested.}
	\label{fig:2}
\end{figure*}

In post-hoc self-testing we consider a scenario where we have self-tested strategy $\tilde{\mathcal{S}}=(\ket{\tilde{\psi}},\{\tilde{A}^{(j)}_{x}\}_x$, $\{\tilde{B}^{(k)}_{y}\}_y)$, and we would like to self-test an additional measurement $\{\tilde{O}^{(\ell)}\}$. We are interested to ask when can $\{\tilde{O}^{(\ell)}\}$ be self-tested by extending $\tilde{\mathcal{S}}$. In particular, when is $\tilde{\mathcal{S}}'=(\ket{\tilde{\psi}},\{\tilde{A}^{(j)}_{x}\}_x,\{\tilde{B}^{(k)}_{y},\tilde{O}^{(\ell)}\}_y)$ self-tested by the correlation it produces (Fig. \ref{fig:2})?

Since the reference strategy $\tilde{\mathcal{S}}=(\ket{\tilde{\psi}}_{\tilde{A}\tilde{B}},\{\tilde{A}^{(j)}_{x}),\{\tilde{B}^{(k)}_{y}\})$ is robustly self-tested from correlation, Alice has to honestly perform some measurement which is an approximate local dilation of $\{\tilde{A}^{(j)}_{x}\}$: 
\begin{align}
    &(\ket{\psi}_{AB},A^{(j)}_{x}\otimes I_B)\xhookrightarrow[\varepsilon]{\Phi,\ket{\operatorname{aux}}}(\ket{\tilde{\psi}}_{\tilde{A}\tilde{B}},\tilde{A}^{(j)}_{x}\otimes I_{\tilde{B}})
\end{align}

Then post-hoc self-testing for an additional observable $\tilde{O}^{(l)}$ would ask that the \emph{same} $\Phi$ and $\ket{\operatorname{aux}}$ also connect $O^{(l)}$ and $\tilde{O}^{(l)}$ on the \emph{same} shared state for any $O^{(l)}$ generating correlation close to that of $\tilde{O}^{(l)}$.

\begin{definition}[robust post-hoc self-testing]
Given the state $\ket{\tilde{\psi}}=\ket{\tilde{\psi}}_{\tilde{A}\tilde{B}}$ and generalized observables $\{\tilde{A}^{(j)}_{x}\}$, a $L$-output generalized observable $\{\tilde{O}^{(l)}\}$ is \emph{robustly post-hoc self-tested} (by the correlation $\{\braket{\tilde{\psi}|\tilde{A}^{(j)}_x\otimes\tilde{O}^{(l)}|\tilde{\psi}}\}$) based on $(\ket{\tilde{\psi}}_{\tilde{A}\tilde{B}},\{\tilde{A}^{(j)}_{x}\})$ if the following condition holds: for any $\varepsilon'>0$, there exist $\varepsilon>0$ and $\delta>0$ such that:

if $(\ket{\psi}_{AB},A^{(j)}_{x}\otimes I_B)\xhookrightarrow[\varepsilon]{\Phi,\ket{\operatorname{aux}}}(\ket{\tilde{\psi}}_{\tilde{A}\tilde{B}},\tilde{A}^{(j)}_{x}\otimes I_{\tilde{B}}\}$ for state $\ket{\psi}_{AB}$, generalized observables $\{{A}^{(j)}_{x}\}$, local isometry $\Phi=\Phi_A\otimes\Phi_B$ and state $\ket{\operatorname{aux}}_{A'B'}$, then any generalized observable $\{O^{(l)}\}$ $\delta$-approximately generating the correlation (\ie, $|\braket{\psi|A^{(j)}_x\otimes O^{(l)}|\psi}-\braket{\tilde{\psi}|\tilde{A}^{(j)}_x\otimes \tilde{O}^{(l)}_y|\tilde{\psi}}|<\delta$) satisfies
$$
(\ket{\psi}_{AB},I_A\otimes O^{(l)})\xhookrightarrow[\varepsilon']{\Phi,\ket{\operatorname{aux}}}(\ket{\tilde{\psi}}_{\tilde{A}\tilde{B}},I_{\tilde{A}}\otimes\tilde{O}^{(l)}),
$$
for all $l\in[0,L-1]$.
\label{def:posthocrob}
\end{definition}

Then post-hoc self-testing extends self-testing protocol in the following sense:
\begin{proposition}
If correlation $\{\braket{\tilde{\psi}|\tilde{A}^{(j)}_x\otimes \tilde{B}^{(k)}_y|\tilde{\psi}}\}$ robustly self-tests 
$\tilde{\mathcal{S}}=(\ket{\tilde{\psi}}_{\tilde{A}\tilde{B}},\{\tilde{A}^{(j)}_{x}),\{\tilde{B}^{(k)}_{y}\})$, and correlation $\{\braket{\tilde{\psi}|\tilde{A}^{(j)}_x\otimes \tilde{O}^{(k)}|\tilde{\psi}}\}$ robustly post-hoc self-tests $\{\tilde{O}^{(l)}\}$ based on $(\ket{\tilde{\psi}}_{\tilde{A}\tilde{B}},\{\tilde{A}_{x}^{(j)}\})$, then the extended correlation $\{\braket{\tilde{\psi}|\tilde{A}^{(j)}_x\otimes \tilde{B}^{(k)}_y|\tilde{\psi}}\}\cup\{\braket{\tilde{\psi}|\tilde{A}^{(j)}_x\otimes \tilde{O}^{(l)}|\tilde{\psi}}\}$ robustly self-tests the extended strategy $\tilde{\mathcal{S}}^{Extend}=(\ket{\tilde{\psi}}_{\tilde{A}\tilde{B}},\{\tilde{A}^{(j)}_{x}\},\{\tilde{B}^{(k)}_{y},\tilde{O}^{(l)}\})$.
\label{prop:extend}
\end{proposition}

\begin{proof}
    By robust post-hoc self-testing, for any $\varepsilon_1$ there exist $\varepsilon_2$ and $\delta_1$ such that 
    $$
    (\ket{\psi}_{AB},I_A\otimes O^{(l)})\xhookrightarrow[\varepsilon_1]{\Phi,\ket{\operatorname{aux}}}(\ket{\tilde{\psi}}_{\tilde{A}\tilde{B}},I_{\tilde{A}}\otimes\tilde{O}^{(l)}),
    $$
    if $|\braket{\psi|A^{(j)}_x\otimes O^{(l)}|\psi}-\braket{\tilde{\psi}|\tilde{A}^{(j)}_x\otimes \tilde{O}^{(l)}_y|\tilde{\psi}}|<\delta_1$. Since $\tilde{\mathcal{S}}=(\ket{\tilde{\psi}}_{\tilde{A}\tilde{B}},\{\tilde{A}^{(j)}_{x}\},\{\tilde{B}^{(k)}_{y}\})$ is robustly self-tested, for $\varepsilon_2$ there exist $\delta_2$ such that 
    \begin{align*}
    &(\ket{\psi}_{AB},A^{(j)}_{x}\otimes I_B)\xhookrightarrow[\varepsilon_2]{\Phi,\ket{\operatorname{aux}}}(\ket{\tilde{\psi}}_{\tilde{A}\tilde{B}},\tilde{A}^{(j)}_{x}\otimes I_{\tilde{B}}),\\
    &(\ket{\psi}_{AB},I_A\otimes B^{(k)}_{y})\xhookrightarrow[\varepsilon_2]{\Phi,\ket{\operatorname{aux}}}(\ket{\tilde{\psi}}_{\tilde{A}\tilde{B}},I_{\tilde{A}}\otimes\tilde{B}^{(k)}_{y}),
\end{align*}
if $|\braket{\psi|A_x^{(j)}\otimes B^{(k)}_y|\psi}- \braket{\tilde{\psi}|\tilde{A}^{(j)}_x\otimes \tilde{B}^{(k)}_y|\tilde{\psi}}|<\delta_2$. Take $\delta=\min\{\delta_1,\delta_2\}$ and $\varepsilon=\max\{\varepsilon_1,\varepsilon_2\}$ one get that the extended strategy $\tilde{\mathcal{S}}^{Extend}=(\ket{\tilde{\psi}}_{\tilde{A}\tilde{B}},\{\tilde{A}^{(j)}_{x}\},\{\tilde{B}^{(k)}_{y},\tilde{O}^{(l)}\})$ is robustly self-tested.
\end{proof}

We visualize the extension of the correlation table to help better understand post-hoc self-testing. For simplicity, consider binary observables $\tilde{A}_x,\tilde{B}_y,\tilde{O}$. The correlation generated by $\tilde{\mathcal{S}}$ can be written as a $(|\mathcal{I}_A|+1)\times(|\mathcal{I}_B|+1)$ table as in Table \ref{table:a}. Then we say that Table \ref{table:a} self-tests strategy $\tilde{\mathcal{S}}$. Take $\tilde{\mathcal{S}}$ as the initial strategy, then add an additional binary observable $\tilde{O}$ on Bob's side; this will extend the correlation table as in Table \ref{table:b}. Intuitively, given self-tested $\{\tilde{A}_x\}$, then for some $\tilde{O}$ it could be the case that $\tilde{O}$ is fully characterized by $\braket{I\otimes \tilde{O}}$ and $\{\braket{\tilde{A}_x\otimes \tilde{O}}\}_x$. If so, we say that $\tilde{O}$ is post-hoc self-tested based on $\{\tilde{A}_x\}$ and $\ket{\tilde{\psi}}$. Then the extended Table \ref{table:b} self-tests $\tilde{\mathcal{S}}^{Extend}$, because essentially the white part of Table \ref{table:b} tests $\tilde{\mathcal{S}}$, and the yellow part tests $\tilde{O}$.

\begin{table}[ht]
\small
\centering
\begin{tabular}{c|c|c|c|c|}
 & $I$ & $\tilde{B}_0$ & ... & $\tilde{B}_{Y-1}$ \\
\hline
$I$ & - & $\braket{I, \tilde{B}_0}_{\tilde{\psi}}$ & ... & $\braket{I, \tilde{B}_{Y-1}}_{\tilde{\psi}}$ \\
\hline
$\tilde{A}_0$ & $\braket{\tilde{A}_0, I}_{\tilde{\psi}}$ & $\braket{\tilde{A}_0, \tilde{B}_0}_{\tilde{\psi}}$ & ... & $\braket{\tilde{A}_0, \tilde{B}_{Y-1}}_{\tilde{\psi}}$ \\
\hline
... & ... & ... & ... & ... \\
\hline
$\tilde{A}_{X-1}$ & $\braket{\tilde{A}_{X-1}, I}_{\tilde{\psi}}$ & $\braket{\tilde{A}_{X-1}, \tilde{B}_0}_{\tilde{\psi}}$ & ... & $\braket{\tilde{A}_{X-1}, \tilde{B}_{Y-1}}_{\tilde{\psi}}$ \\
\hline
 \end{tabular}
\caption{Initiate correlation table. Here $\braket{\tilde{A}, \tilde{B}}_{\tilde{\psi}}$ is in short for $\braket{\tilde{\psi}|\tilde{A}\otimes \tilde{B}|\tilde{\psi}}$, and we take $X=|\mathcal{I}_A|$, $Y=|\mathcal{I}_B|$.
}
\label{table:a}
\end{table} 

\begin{table}[ht]
\small
\centering
\begin{tabular}{c|c|c|c|c|c|}
 & $I$ & $\tilde{B}_0$ & ... & $\tilde{B}_{Y-1}$ & \cellcolor{yellow}{$\tilde{O}$} \\
\hline
$I$ & - & $\braket{I, \tilde{B}_0}_{\tilde{\psi}}$ & ... & $\braket{I, \tilde{B}_{Y-1}}_{\tilde{\psi}}$ & \cellcolor{yellow}{$\braket{I, \tilde{O}}_{\tilde{\psi}}$} \\
\hline
$\tilde{A}_0$ & $\braket{\tilde{A}_0, I}_{\tilde{\psi}}$ & $\braket{\tilde{A}_0, \tilde{B}_0}_{\tilde{\psi}}$ & ... & $\braket{\tilde{A}_0, \tilde{B}_{Y-1}}_{\tilde{\psi}}$ & \cellcolor{yellow}{$\braket{\tilde{A}_0, \tilde{O}}_{\tilde{\psi}}$} \\
\hline
... & ... & ... & ... & ... & \cellcolor{yellow}{...} \\
\hline
$\tilde{A}_{X-1}$ & $\braket{\tilde{A}_{X-1}, I}_{\tilde{\psi}}$ & $\braket{\tilde{A}_{X-1}, \tilde{B}_0}_{\tilde{\psi}}$ & ... & $\braket{\tilde{A}_{X-1}, \tilde{B}_{Y-1}}_{\tilde{\psi}}$ & \cellcolor{yellow}{$\braket{\tilde{A}_{X-1}, \tilde{O}}_{\tilde{\psi}}$} \\
\hline
\end{tabular}
\caption{Extended correlation table.}
\label{table:b}
\end{table} 

From hereon we shall call $\tilde{\mathcal{S}},\ket{\tilde{\psi}}_{\tilde{A}\tilde{B}},\{\tilde{A}_{x}^{(j)}\}$ the initial strategy, initial state, and initial generalized observables, respectively, and call $\tilde{O}^{(l)}$ the additional generalized observables.

\subsection{Robust post-hoc self-testing criterion for projective strategies}

Given the set of initial generalized observables $\{\tilde{A}_{x}\}$ together with the initial state $\ket{\tilde{\psi}}$, what kind of generalized observable $\tilde{O}$ is post-hoc self-tested based on $(\ket{\tilde{\psi}},\tilde{A}_{x})$? Intuitively, if $\{\braket{\tilde{\psi}|(\tilde{A}^{(j)}_x\otimes\tilde{O}^{(\ell)})|\tilde{\psi}}\}_x$ can fully characterize $\{\tilde{O}^{(\ell)}\}$ for all $\ell$ then Bob also has no choice but to honestly perform a local dilation of $\tilde{O}^{(\ell)}$ on $\ket{\tilde{\psi}}$. This then gives a criterion of post-hoc self-testing. In proving this criterion, a version of the folklore fact `any vector is uniquely determined by its inner products with basis vectors' is useful. Explicitly,

\begin{lemma}
Let $\tilde{v}_0,\dots,\tilde{v}_{n-1}$ be linearly independent vectors in a Hilbert space. 
Let ${v}_0,\dots,{v}_{n-1}$ be nearby vectors  (in the norm induced by the inner product $\|a\|=\sqrt{\braket{a,a}}$),
$$
\forall x\in[0,n-1],\|v_x-\tilde{v}_x\|<\varepsilon.
$$
For any vector pair $v$ and $\tilde{v}$ such that $\braket{v,v}\le\braket{\tilde{v},\tilde{v}}$ and $\tilde{v}\in\operatorname{span}_{\mathbb{C}}\{v_0,\dots,v_{n-1}\}$, if
\begin{align*}
    \forall x\in[0,n-1],|\braket{v_x,v}-\braket{\tilde{v}_x,\tilde{v}}|<\delta
\end{align*}
then 
$$
\|v-\tilde{v}\|\le\left(\frac{4n}{\lambda_{\min}(G)}\right)^{\frac{1}{4}}\left(\varepsilon\|\tilde{v}\|+\delta\right)^{\frac{1}{2}}\|\tilde{v}\|^{\frac{1}{2}},
$$
where $G$ is the Gram matrix of $\tilde{v}_0,\dots,\tilde{v}_{n-1}$ with entries $g_{jk}=\braket{\tilde{v}_j,\tilde{v}_k}$, and $\lambda_{\min}(G)$ is the minimal eigenvalue of $G$.
\label{lem:inspacerobust}
\end{lemma}

\begin{proof}
Since $\tilde{v}\in\operatorname{span}_{\mathbb{C}}\{\tilde{v}_0,\dots,\tilde{v}_{n-1}\}$, let $\tilde{v}=\sum_x\alpha_x\tilde{v}_x=W\alpha$, where
\begin{align*}
    W=\begin{pmatrix}
    | & & |\\
    \tilde{v}_0 & \cdots & \tilde{v}_{n-1}\\
    | & & |
    \end{pmatrix},
\end{align*}
then
\begin{align*}
    \|v-\tilde{v}\|^2=&\|v\|^2+\|\tilde{v}\|^2-2\operatorname{Re}\braket{v,\tilde{v}}\nonumber\\
    =&\|v\|^2+\|\tilde{v}\|^2-2\|\tilde{v}\|^2-2\operatorname{Re}\braket{v-\tilde{v},\tilde{v}}\nonumber\\
    \le&-2\operatorname{Re}\braket{v-\tilde{v},\tilde{v}}\nonumber\\
    \le&2|\braket{v-\tilde{v},\tilde{v}}|\nonumber\\
    =&2|\sum_{x=0}^{n-1}\alpha_x\braket{v-\tilde{v},\tilde{v}_x}|\nonumber\\
    \le&2\sum_{x=0}^{n-1}|\alpha_x|\cdot|\braket{v,\tilde{v}_x}-\braket{\tilde{v},\tilde{v}_x}|\nonumber\\
    =&2\sum_{x=0}^{n-1}|\alpha_x|\cdot|\braket{v,\tilde{v}_x}-\braket{{v},v_x}+\braket{{v},v_x}-\braket{\tilde{v},\tilde{v}_x}|\nonumber\\
    \le&2\sum_{x=0}^{n-1}|\alpha_x|\cdot(\varepsilon\|v\|+\delta)\nonumber\\
    \le&2\|\alpha\|_1(\varepsilon\|\tilde{v}\|+\delta),
\end{align*}
where $\|\cdot\|_1$ is the vector 1-norm. Using the vector norm inequality we have
\begin{align*}
    \|\alpha\|_1\le&\sqrt{n}\|\alpha\|\nonumber\\
    =&\sqrt{n}\|(W^\dagger W)^{-1}W^\dagger W\alpha\|\nonumber\\
    \le&\sqrt{n}\|(W^\dagger W)^{-1}W^\dagger\|_\infty\|W\alpha\|,
\end{align*}
where $\|\cdot\|_\infty$ is the spectral norm of operators (Schatten $\infty$-norm). 

Note that $W$ admits a singular value decomposition $W=V\Sigma U^\dagger$, where $V$ is isometry, $U$ is unitary, and $\Sigma=\operatorname{diag}(\sigma_0,\dots,\sigma_{n-1})$ is positive definite. Then $G=W^\dagger W=U\Sigma^2 U^\dagger$. Therefore
\begin{align*}
    \|(W^\dagger W)^{-1}W^\dagger\|_\infty=\|(U\Sigma^2U^\dagger)^{-1}U\Sigma V^\dagger\|_\infty=\|U\Sigma^{-1}V^\dagger\|_\infty=\sigma_{\max}(U\Sigma^{-1}V^\dagger)=\frac{1}{\sqrt{\lambda_{\min}(G)}}.
\end{align*}
Finally, 
\begin{align*}
    \|v-\tilde{v}\|^2\le&2\|\alpha\|_1(\varepsilon\|\tilde{v}\|+\delta)\nonumber\\
    \le&2\sqrt{n}\|(W^\dagger W)^{-1}W^\dagger\|_\infty\|W\alpha\|(\varepsilon\|\tilde{v}\|+\delta)\nonumber\\
    =&\frac{2\sqrt{n}}{\sqrt{\lambda_{\min}(G)}}(\varepsilon\|\tilde{v}\|+\delta)\|\tilde{v}\|\nonumber\\
    \implies\quad\|v-\tilde{v}\|\le&\left(\frac{4n}{\lambda_{\min}(G)}\right)^{\frac{1}{4}}\left(\varepsilon\|\tilde{v}\|+\delta\right)^{\frac{1}{2}}\|\tilde{v}\|^{\frac{1}{2}}.
\end{align*}
\end{proof}

The analogue of Lemma \ref{lem:inspacerobust} for unitary operators is crucial in the following proposition. From hereon we assume that the reference strategy is given in a Schmidt basis for its state.

\begin{proposition}
Let $\ket{\tilde{\psi}}_{\tilde{A}\tilde{B}}\in \mathcal{H}_{\tilde A}\otimes\mathcal{H}_{\tilde B}$ be a state, and $\{\tilde{A}_{x}\}$, $x\in[0,n-1]$ be unitaries in $\mathcal{L}(H_{\tilde{A}})$. Suppose
$\tilde{O}\in\mathcal{L}(H_{\tilde{B}})$ is a unitary 
such that
$$
\Bar{\tilde{O}}P\in\operatorname{span}_{\mathbb{C}}\{D{\tilde{A}}_xD\}
$$
where $P$ is positive definite and $D=\operatorname{vec}^{-1}(\ket{\tilde{\psi}})=\operatorname{diag}(\lambda_0,...,\lambda_{d-1})$, where $\lambda_j$ are Schmidt coefficients of $\ket{\tilde{\psi}}$.

If states $\ket{\psi}_{AB}\in \mathcal{H}_{A}\otimes\mathcal{H}_{B},
\ket{\operatorname{aux}}\in
\mathcal{H}_{A'}\otimes\mathcal{H}_{B'}$, 
contractions $\{{A}_{x}\}$ in $\mathcal{L}(H_{A})$, 
a contraction $O\in\mathcal{L}(\mathcal{H}_{B})$, and a local isometry 
$\Phi=\Phi_A\otimes\Phi_B:
\mathcal{H}_{A}\otimes\mathcal{H}_{B} \to
(\mathcal{H}_{\tilde A}\otimes\mathcal{H}_{A'})\otimes (\mathcal{H}_{\tilde B}\otimes\mathcal{H}_{B'})$ satisfy
\begin{align*}
&\forall x,~(\ket{\psi}_{AB},A_{x}\otimes I_B)\xhookrightarrow[\varepsilon]{\Phi,\ket{\operatorname{aux}}}(\ket{\tilde{\psi}}_{\tilde{A}\tilde{B}},\tilde{A}_{x}\otimes I_{\tilde{B}}),\\
&\ket{\psi}_{AB}\xhookrightarrow[\varepsilon]{\Phi,\ket{\operatorname{aux}}}\ket{\tilde{\psi}}_{\tilde{A}\tilde{B}},\\
&|\braket{\psi|A_x\otimes O|\psi}-\braket{\tilde{\psi}|\tilde{A}_x\otimes \tilde{O}|\tilde{\psi}}|<\delta,
\end{align*}
then $(\ket{\psi}_{AB},{I}_A\otimes O)\xhookrightarrow[\varepsilon']{\Phi,\ket{\operatorname{aux}}}(\ket{\tilde{\psi}}_{\tilde{A}\tilde{B}},I_{\tilde{A}}\otimes\tilde{O})$, where
\begin{align}
\varepsilon'=\left(\frac{n}{\lambda_{\min}(G)}\right)^{\frac{1}{4}}\left(2\frac{\tr Q}{\lambda_{\min}(Q)}\kappa(D)\right)^{\frac{1}{2}}\left(\left(2\left(\frac{\tr Q}{\lambda_{\min}(Q)}\right)^{\frac{1}{2}}\lambda_{\max}(D)+1\right)\varepsilon+\delta\right)^{\frac{1}{2}}+\varepsilon.
\label{eq:delta}
\end{align}
Here $Q=D^{-1}PD^{-1}$, $G$ is the Gram matrix of $\{\tilde{A}_x\}$ with etries $g_{jk}=\tr[\tilde{A}^\dagger_j\tilde{A}_k]$, and $\kappa(D)$ is the condition number of $D$, \ie, the ratio of the maximal and the minimal Schmidt coefficient of $\ket{\tilde{\psi}}$.
\label{proposition:eachkrobust}
\end{proposition}

\begin{proof}
Define
\begin{align}
    &\tilde{v}_x:=D'\otimes\sqrt{P}^{-1}D\tilde{A}_x^\dagger D,\nonumber\\
    &v_x:=(I\otimes\sqrt{P^{-1}}D)(\Phi_A A_x^\dagger\Phi_A^\dagger)(D'\otimes D),\nonumber\\
    &\tilde{v}:=(D'\otimes \sqrt{P})(I_{B'}\otimes\tilde{O}^\intercal)=D'\otimes\sqrt{P}\tilde{O}^\intercal,\nonumber\\
    &v:=(D'\otimes\sqrt{P})(\Phi_B O\Phi_B^\dagger)^\intercal.\nonumber
\end{align}
where $D'=\operatorname{vec}^{-1}(\ket{\operatorname{aux}})$. We also consider $\ket{\operatorname{aux}}$ given in its Schmidt basis, so $D'$ is diagonal (while not necessarily full-ranked). Note that $\tr[D'^2]=\tr[\rho_A]=1$. The entries of the Gram matrix $G'$ for $\{\tilde{v}_x\}$ are $g'_{jk}=\tr[\tilde{v}_j^\dagger\tilde{v}_k]=\tr[(D^{-1}PD^{-1})^{-1}\tilde{A}^\dagger_kD^2\tilde{A}_j]=\tr[(Q)^{-1}\tilde{A}^\dagger_kD^2\tilde{A}_j]$. Comparing the minimal eigenvalues of $G$ and $G'$, we have that
$$
\lambda_{\min}(G')\ge\lambda_{\min}(G)\lambda_{\min}(D^2)\lambda_{\min}(Q^{-1})=\frac{\lambda_{\min}(G)\lambda^2_{\min}(D)}{\lambda_{\max}(Q)}>\frac{\lambda_{\min}(G)\lambda^2_{\min}(D)}{\tr Q}.
$$

To apply Lemma \ref{lem:inspacerobust}, one check the conditions:
\begin{itemize}
    \item $\tilde{v}\in\operatorname{span}_{\mathbb{C}}\{\tilde{v}_x\}$:
    \begin{align*}
    &\Bar{\tilde{O}}P\in\operatorname{span}_{\mathbb{C}}\{D{\tilde{A}}_xD\}\nonumber\\
    \Rightarrow\quad&P(\tilde{O})^\intercal\in\operatorname{span}_{\mathbb{C}}\{D\tilde{A}_x^\dagger D\}\nonumber\\
    \Rightarrow\quad&D'\otimes\sqrt{P}(\tilde{O})^\intercal\in\operatorname{span}_{\mathbb{C}}\{D'\otimes\sqrt{P}^{-1}D\tilde{A}_x^\dagger D\}\nonumber\\
    \Rightarrow\quad&\tilde{v}\in\operatorname{span}_{\mathbb{C}}\{\tilde{v}_x\}.
\end{align*}

    \item $\|v\|\le\|\tilde{v}\|$:
    \begin{align*}
    \|\tilde{v}\|=&\sqrt{\tr[D'^2\otimes\sqrt{P}(\tilde{O})^\intercal((\tilde{O})^\intercal)^\dagger\sqrt{P}]}=\sqrt{\tr[P]},\\
    \|v\|=&\sqrt{\tr[(\Phi_B O\Phi_B^\dagger)^\intercal(D'\otimes\sqrt{P})(D'\otimes\sqrt{P})((\Phi_B O\Phi_B^\dagger)^\intercal)^\dagger]}\nonumber\\
    =&\sqrt{\tr[(\Phi_B O\Phi_B^\dagger)^\dagger(D'^2\otimes\overline{P})(\Phi_B O\Phi_B^\dagger)]}\nonumber\\
    =&\sqrt{\tr[(D'^2\otimes\overline{P})\Phi_B O\Phi_B^\dagger\Phi_B (O)^\dagger\Phi_B^\dagger]}\nonumber\\
    \le&\sqrt{\tr[(D'^2\otimes\overline{P})\Phi_B\Phi_B^\dagger]}\nonumber\\
    \le&\sqrt{\tr[(D'^2\otimes\overline{P})]}\nonumber\\
    =&\sqrt{\tr[{P}]}=\|\tilde{v}\|,
\end{align*}
    where the first inequality comes from $O$ being a contraction, and the second inequality comes from $\Phi_B\Phi_B^\dagger\le I_{\tilde{B}B'}$ and $D'^{2}\otimes\overline{P}\ge0$.

    \item for all $x$, $v_x$ and $\tilde{v}_x$ are close:
    \begin{align*}
        \|{v}_x-\tilde{v}_x\|\le&\|(\Phi_A A_x^\dagger\Phi_A^\dagger)(D'\otimes D)-D'\otimes\tilde{A}_x^\dagger D\|\|\sqrt{P^{-1}}D\|_\infty\nonumber\\
        =&\|(D'\otimes D)(\Phi_A A_x^\dagger\Phi_A^\dagger)-D'\otimes D\tilde{A}_x^\dagger\|\|\sqrt{P^{-1}}D\|_\infty\nonumber\\
        =&\|(\Phi_A A_x\Phi_A^\dagger)(D'\otimes D)-D'\otimes\tilde{A}_xD\|\|\sqrt{P^{-1}}D\|_\infty\nonumber\\
        =&\|(\Phi_A A_x\Phi_A^\dagger\otimes I)(\ket{\operatorname{aux}}\otimes\ket{\tilde{\psi}})-\ket{\operatorname{aux}}\otimes(\tilde{A}_x\otimes I\ket{\tilde{\psi}})\|\|\sqrt{P^{-1}}D\|_\infty\nonumber\\
        \le&(\|(\Phi_A A_x\Phi_A^\dagger\otimes I)\Phi[\ket{\psi}]-\ket{\operatorname{aux}}\otimes(\tilde{A}_x\otimes I\ket{\tilde{\psi}})\|+\varepsilon)\|\|\sqrt{P^{-1}}D\|_\infty\nonumber\\
        =&(\|\Phi [A_x\otimes I\ket{\psi}]-\ket{\operatorname{aux}}\otimes(\tilde{A}_x\otimes I\ket{\tilde{\psi}})\|+\varepsilon)\|\sqrt{P^{-1}}D\|_\infty\nonumber\\
        =&\frac{2\varepsilon}{\lambda_{\min}(D^{-1}PD^{-1})^{\frac{1}{2}}}=\frac{2\varepsilon}{\lambda_{\min}(Q)^{\frac{1}{2}}}.
    \end{align*}
    
    \item the inner products are close:
\begin{align*}
    &|\braket{v_x,v}-\braket{\tilde{v}_x,\tilde{v}}|\nonumber\\
    =&|\braket{(I\otimes \sqrt{P^{-1}}D)(\Phi_A A^\dagger_x\Phi_A^\dagger)(D'\otimes D),(D'\otimes\sqrt{P})(\Phi_B O\Phi_B^\dagger)}-\braket{\tilde{v}_x,\tilde{v}}|\nonumber\\
    =&|\braket{D'\otimes D,(\Phi_A A_x\Phi_A^\dagger)(I\otimes D\sqrt{P^{-1}})(D'\otimes\sqrt{P})(\Phi_B O\Phi_B^\dagger)}-\braket{\tilde{v}_x,\tilde{v}}|\nonumber\\
    =&|\braket{D'\otimes D,(\Phi_A A_x\Phi_A^\dagger)(D'\otimes D)(\Phi_B O\Phi_B^\dagger)}-\braket{\tilde{v}_x,\tilde{v}}|\nonumber\\
    =&|\braket{\ket{\operatorname{aux}}\otimes\ket{\tilde{\psi}},(\Phi_A A_x\Phi_A^\dagger\otimes\Phi_B O\Phi_B^\dagger)(\ket{\operatorname{aux}}\otimes\ket{\tilde{\psi}})}-\braket{\tilde{v}_x,\tilde{v}}|\nonumber\\
    \le&|\braket{\Phi[\ket{\psi}],(\Phi_A A_x\Phi_A^\dagger\otimes\Phi_B O\Phi_B^\dagger)\Phi[\ket{\psi}]}-\braket{\tilde{v}_x,\tilde{v}}|+2\varepsilon\nonumber\\
    =&|\braket{\ket{\psi},A_x\otimes O\ket{\psi}}-\braket{\tilde{v}_x,\tilde{v}}|+2\varepsilon\nonumber\\
    =&|\braket{\psi|A_x\otimes O|\psi}-\braket{\tilde{\psi}|\tilde{A}_x\otimes \tilde{O}|\tilde{\psi}}|+\varepsilon<\delta+2\varepsilon.
\end{align*}
\end{itemize}

So all the conditions of Lemma \ref{lem:inspacerobust} hold. By Lemma \ref{lem:inspacerobust}, we have
\begin{align*}
    \|v-\tilde{v}\|\le&\left(\frac{4n}{\lambda_{\min}(G')}\right)^{\frac{1}{4}}\left(\frac{2}{\lambda_{\min}(Q)^{\frac{1}{2}}}\varepsilon(\tr P)^{\frac{1}{2}}+\delta+2\varepsilon\right)^{\frac{1}{2}}(\tr P)^{\frac{1}{4}}\nonumber\\
    =&\left(\frac{4n\tr(P)\tr(Q)}{\lambda_{\min}(G)\lambda_{\min}(D)^2}\right)^{\frac{1}{4}}\left(\left(2\left(\frac{\tr P}{\lambda_{\min}(Q)}\right)^{\frac{1}{2}}+2\right)\varepsilon+\delta\right)^{\frac{1}{2}}\nonumber\\
    =&\left(\frac{4n(\tr Q)^2\lambda_{\max}(D)^2}{\lambda_{\min}(G)\lambda_{\min}(D)^2}\right)^{\frac{1}{4}}\left(\left(2\left(\frac{\tr Q}{\lambda_{\min}(Q)}\right)^{\frac{1}{2}}\lambda_{\max}(D)+2\right)\varepsilon+\delta\right)^{\frac{1}{2}}\nonumber\\
    =&\left(\frac{n}{\lambda_{\min}(G)}\right)^{\frac{1}{4}}(2(\tr Q)\kappa(D))^{\frac{1}{2}}\left(\left(2\left(\frac{\tr Q}{\lambda_{\min}(Q)}\right)^{\frac{1}{2}}\lambda_{\max}(D)+2\right)\varepsilon+\delta\right)^{\frac{1}{2}},
\end{align*}
which implies
\begin{align*}
    &\|\Phi[I\otimes O\ket{\psi}]-\ket{\operatorname{aux}}\otimes (I_{\tilde{A}}\otimes\tilde{O}\ket{\tilde{\psi}})\|\nonumber\\
    \le&\|I_{\tilde{A}A'}\otimes\Phi_B O\Phi_B^\dagger(\ket{\operatorname{aux}}\otimes\ket{\tilde{\psi}})-\ket{\operatorname{aux}}\otimes (I_{\tilde{A}}\otimes \tilde{O}\ket{\tilde{\psi}})\|+\varepsilon\nonumber\\
    =&\|(D'\otimes D)(\Phi_B O\Phi_B^\dagger)^\intercal-(D'\otimes D)(I_{B'}\otimes\tilde{O}^\intercal)\| +\varepsilon\nonumber\\
    \le&\|v-\tilde{v}\|\|\sqrt{P^{-1}}D\|_\infty+\varepsilon\nonumber\\
    \le&\left(\frac{n}{\lambda_{\min}(G)}\right)^{\frac{1}{4}}\left(2\frac{\tr Q}{\lambda_{\min}(Q)}\kappa(D)\right)^{\frac{1}{2}}\left(\left(2\left(\frac{\tr Q}{\lambda_{\min}(Q)}\right)^{\frac{1}{2}}\lambda_{\max}(D)+2\right)\varepsilon+\delta\right)^{\frac{1}{2}}+\varepsilon.
\end{align*}
\end{proof}

A few remarks of Proposition \ref{proposition:eachkrobust}:
\begin{enumerate}
    \item If we fix $P=I$, then the criterion of Proposition \ref{proposition:eachkrobust} reduces to $\Bar{\tilde{O}}\in\operatorname{span}\{D\tilde{A}_xD\}$, which is foreseeable from the fact that $\braket{\tilde{\psi}|\tilde{A}_x\otimes \tilde{O}|\tilde{\psi}}=\tr[D\tilde{A}_xD\tilde{O}^\intercal]=\braket{D\tilde{A}_x^\dagger D,\tilde{O}^\intercal}$. Our result however, allows for more general $\tilde{O}$ than just the linear combinations of $\{D\tilde{A}_xD\}$.
    \item For small $\varepsilon,\delta$ we have $\varepsilon'=O(\sqrt{C\varepsilon+\delta})$. If the initial strategy has explicit $\varepsilon-\delta$ dependence, by Proposition \ref{proposition:eachkrobust} the extended strategy will also have explicit robustness.
    \item In the mirror case where we have additional unitary $\tilde{O}$ on Alice's side and Bob's unitaries are $\tilde{B}_y$, the criterion is similar:
\begin{align*}
\Bar{\tilde{O}}P_j\in\operatorname{span}_{\mathbb{C}}\{D{\tilde{B}}_yD\}.
\end{align*}
    \item In Eq. \eqref{eq:delta}, $\kappa(D)$ and $\lambda_{\max}(D)$ imply that more entanglement enables more robustness, which is intuitive: imagine that $\ket{\tilde{\psi}}$ is weakly entangled (which leads to a large $\kappa(D)$), then Alice and Bob are so weakly correlated that we cannot control $O$ from $\{A_x^{(j)}\}$.
\end{enumerate}

Now we are ready to provide a sufficient condition for $\tilde{O}$ being post-hoc self-tested based on $(\ket{\tilde{\psi}},\tilde{A}_{x}\}$. If the condition in Proposition \ref{proposition:eachkrobust} is satisfied for all powers of a generalized observable $\tilde{O}$ as required by Definition \ref{def:local_dilationobs}, we immediately have the following criterion:

\begin{theorem}
An additional $L$-output projective measurement, characterized by observable $\{\tilde{O}\}$, is robustly post-hoc self-tested based on a robustly self-tested initial observables $\{\tilde{A}_{x}\}$ and initial state $\ket{\tilde{\psi}}_{\tilde{A}\tilde{B}}$, if there exist positive definite operators $P_l>0$ such that
\begin{align*}
\Bar{\tilde{O}}^{l}P_l\in\operatorname{span}_{\mathbb{C}}\{D{\tilde{A}}^{j}_xD:x,j\},
\end{align*}
for every $l\in[0,L-1]$. Here $D=\operatorname{vec}^{-1}(\ket{\tilde{\psi}})=\operatorname{diag}(\lambda_1,...,\lambda_{d})$, where $\lambda_j$ are Schmidt coefficients of $\ket{\tilde{\psi}}$. Moreover, the $(\varepsilon',(\varepsilon,\delta))$ dependence of the robustness will be $\varepsilon'=O(\sqrt{C\varepsilon+\delta})$.
\label{thm:conditionUnifiedrob}
\end{theorem}
\begin{proof}
For every $l\in[0,L-1]$, note that $\tilde{A}_x^{j}$, $\tilde{O}^{l}$ are unitaries, $A_x^{(j)}$, $O^{(l)}$ are contractions, so we can apply Proposition \ref{proposition:eachkrobust} to get $\varepsilon'_l=O(\sqrt{C\varepsilon+\delta})$ by Eq. \eqref{eq:delta}. Taking $\varepsilon'=\max_l\{\varepsilon'_l\}=O(\sqrt{C\varepsilon+\delta})$ then gives the desired conclusion.
\end{proof}

Given concrete $\ket{\tilde{\psi}},\{\tilde{A}_x),\tilde{O}$, the condition $\tilde{O}^{l}P_l\in\operatorname{span}_{\mathbb{C}}\{D{\tilde{A}}^{j}_xD\}$ can be determined via a feasibility semidefinite program (SDP). Moreover, since $P_l$ has the freedom in scaling and $Q_l=D^{-1}P_lD^{-1}$ is positive definite, we can without loss of generality take $\lambda_{\min}(Q_l)=1$, and minimize $\tr Q_l$ by the following SDP to get a better robustness:
\begin{align*}
    \min\quad&\tr Q_l\nonumber\\
    \text{s. t.}\quad&\Bar{\tilde{O}}^{l}DQ_lD=\sum_{j,x}c_{j,x,l}D{\tilde{A}}^{j}_xD,\nonumber\\
    &Q_l\ge I,\nonumber\\
    &c_{j,x,l}\in\mathbb{C}\nonumber
\end{align*}
for every $k$ individually. Also note that Theorem \ref{thm:conditionUnifiedrob} does not assume measurements to be real, so it works for observables of complex reference measurements as well.

\subsection{A closed-form criterion for binary observables}

While the condition in Theorem \ref{thm:conditionUnifiedrob} can be checked through a semidefinite program, the existential nature of it can make it cumbersome to work with in some applications. In order to address this issue, we present a closed-form variant of Theorem \ref{thm:conditionUnifiedrob} for the special case where $\tilde{A}_{x}$ and $\tilde{O}$ are binary measurements. This particular form not only facilitates the proof of our main theorem, but also proves useful in the context of iterative self-testing.

Let all the measurements in $\tilde{\mathcal{S}}$ be binary, \ie, $|\mathcal{O}_A|=|\mathcal{O}_B|=2$. Since $\tilde{A}_x$ and $\tilde{O}$ are now orthogonal matrices (as the projections are real), the condition from Theorem \ref{thm:conditionUnifiedrob} simplifies to
$$
    \tilde{O}P\in\operatorname{span}_{\mathbb{C}}\{D^2,D\tilde{A}_xD\}.
$$
(Note that $\tilde{O}^0P_0=P_0$ is always in the span by taking $P_0=D^{2}$.) Further, we can restrict ourselves in the real span of $\{D^2,D\tilde{A}_xD\}$: if $\tilde{O}P\in\operatorname{span}_{\mathbb{C}}\{D^2,D\tilde{A}_xD\}$, then $\tilde{O}\operatorname{Re}(P)\in\operatorname{span}_{\mathbb{R}}\{D^2,D\tilde{A}_xD\}$ where $\operatorname{Re}(P)$ is positive definite\footnote{$\operatorname{Re}(P)=\frac12(P+\Bar{P})$, where $P$ and $\Bar{P}$ are both positive definite.}. Thus it suffices to consider
\begin{equation}\label{eq:realspan}
    \tilde{O}P\in\operatorname{span}_{\mathbb{R}}\{D^2,D\tilde{A}_xD\},
\end{equation}
where $P$ is real and positive definite. 

Since every operator contained in $\operatorname{span}_{\mathbb{R}}\{D^2,D\tilde{A}_xD\}$ is real Hermitian (or symmetric), consider the following $\operatorname{sgn}$ map that takes real Hermitian matrices to real Hermitian matrices with eigenvalues $0,\pm1$, defined by
\begin{align}
\operatorname{sgn}:{H}(\mathbb{R})_d&\rightarrow{H}(\mathbb{R})_d\nonumber\\
H=\sum_j\lambda_j\proj{v_j}&\mapsto\operatorname{sgn}(H)=\sum_j\operatorname{sgn}(\lambda_j)\proj{v_j}\nonumber
\end{align}
where $(\ket{v_j}\}_j$ is an orthonormal basis of eigenvectors for $H$. That is, $\operatorname{sgn}$ is the extension of the sign function via functional calculus.
Then we show that the criterion Eq. \eqref{eq:realspan} is equivalent to that $\tilde{O}$ is in the image of $\operatorname{span}\{D^2,D\tilde{A}_xD\}$ via $\operatorname{sgn}$:
\begin{lemma}
Given $d$-dimensional orthogonal matrices $\tilde{O}$ and $\{\tilde{A}_x\}$, and $D=\operatorname{diag}(\{\lambda_j\})$ where $\lambda_j>0$ for $j\in[0,d-1]$. Then there exist a real positive definite $P$ such that
\begin{align*}
    \tilde{O}P\in\operatorname{span}_{\mathbb{R}}\{D^2,D\tilde{A}_xD),
\end{align*}
if and only if
\begin{align*}
    \tilde{O}\in\operatorname{sgn}(\operatorname{span}_{\mathbb{R}}\{D^2,D\tilde{A}_xD\}).
\end{align*}
\end{lemma}

\begin{proof}
The `if' part: Let $\tilde{O}=\operatorname{sgn}(H)$ where $H\in\operatorname{span}_{\mathbb{R}}\{D^2,D\tilde{A}_xD\}$. Since $\tilde{O}$ is non-singular, $H$ is also non-singular. Then $\tilde{O}H=\operatorname{sgn}(H)H$ is positive definite. Take $P=\tilde{O}H$ then $\tilde{O}P=H\in\operatorname{span}_{\mathbb{R}}\{D^2,D\tilde{A}_xD\}$.

The `only if' part: Let $\tilde{O}P=H\in\operatorname{span}_{\mathbb{R}}\{D^2,D\tilde{A}_xD\}$, then $H=H^\intercal=(\tilde{O}P)^\intercal=P\tilde{O}$. So $\tilde{O}$, $H$, and $P$ commute, therefore are simultaneously diagonalizable. Let $\{b_j\}$, $\{p_j\}$, $\{h_j\}$ be the eigenvalues of $O,P,H$, respectively; then $o_jp_j=h_j\neq0$. Also note that $o_j=\pm1$ and $p_j>0$, so $p_j=|h_j|$ and $o_j=h_j/|h_j|=\operatorname{sgn}(h_j)$. Therefore $\tilde{O}=\operatorname{sgn}(H)$.
\end{proof}

And the equivalent criterion for post-hoc selt-testing binary observables follows immediately:

\begin{proposition}
An additional binary (2-output) $d$-dimensional observable $\tilde{O}$ is robustly post-hoc self-tested based on robustly self-tested initial binary observables $\{\tilde{A}_{x}\}$ and initial state $\ket{\tilde{\psi}}_{\tilde{A}\tilde{B}}$, if 
\begin{align*}
    \tilde{O}\in\operatorname{sgn}(\operatorname{span}_{\mathbb{R}}\{D^2,D\tilde{A}_xD:x\}),
\end{align*}
where $D=\operatorname{vec}^{-1}(\ket{\tilde{\psi}})$, and $\operatorname{sgn}$ maps real Hermitian matrices to real Hermitian matrices, defined by
\begin{align}
\operatorname{sgn}:{H}(\mathbb{R})_d&\rightarrow{H}(\mathbb{R})_d\nonumber\\
H=\sum_j\lambda_j\proj{v_j}&\mapsto\operatorname{sgn}(H)=\sum_j\operatorname{sgn}(\lambda_j)\proj{v_j}.\nonumber
\end{align}
Moreover, the $\varepsilon'-(\varepsilon,\delta)$ dependence of the robustness will be $\varepsilon'=O(\sqrt{C\varepsilon+\delta})$.
\label{prop:conditionbinaryrob}
\end{proposition}

\section{Iterative self-testing I: self-testing of arbitrary real projective measurements}
\label{sec:d+1}

Now we introduce the technique of iterative self-testing, by which we show how to self-test arbitrary real projective measurements. From now we restrict to reference strategies with binary observables and a maximally entangled initial state $\ket{\tilde{\psi}}=\ket{\Phi_d}=\sum_j\ket{jj}/\sqrt{d}$. In this case, the criterion in Proposition \ref{prop:conditionbinaryrob} reduces to $\tilde{O}\in\operatorname{sgn}(\operatorname{span}\{I,\tilde{A}_x\})$, because $D=1/\sqrt{d}I$ is proportional to the identity matrix. 

Given initial strategy $\tilde{\mathcal{S}}=(\Phi_d,\{\tilde{A}_{x}\},\{\tilde{B}_{y}\})$, if we post-hoc self-test $\tilde{O}\in\operatorname{sgn}(\operatorname{span}\{I,\tilde{A}_x\})$ on Bob's side, then we can use $\{\tilde{B}_y,\tilde{O}\}$ to post-hoc self-test another measurement $\tilde{O}'\in\operatorname{sgn}(\operatorname{span}\{I,\tilde{B}_y,\tilde{O}\})$ for Alice. By doing this in several rounds, starting from a small set of observables $\{\tilde{A}_x\}$ we may eventually self-test many additional observables. We call this process \emph{iterative self-testing}.

We visualize the extension of the correlation table to help better understand iterative self-testing. Let the initial binary observables to be $\{\tilde{A}_x),\{\tilde{B}_y\}$, and the initial state $\ket{\tilde{\psi}}=\ket{\Phi_d}$ is maximally entangled. Then the correlation generated Table \ref{table:c} self-tests the initial strategy $\tilde{\mathcal{S}}$. Recall that the condition from Proposition \ref{prop:conditionbinaryrob} reduces to $\tilde{O}\in\operatorname{sgn}(\operatorname{span}\{I,\tilde{A}_x\})$. Now consider an additional binary observable $\tilde{O}$ such that $\tilde{O}\not\in\operatorname{sgn}(\operatorname{span}\{I,\tilde{A}_x\})$ but $\tilde{O}\in \operatorname{sgn}(\operatorname{span}\{\operatorname{sgn}(\operatorname{span}\{I,\tilde{A}_x\})\})$. Since $\tilde{O}\not\in \operatorname{sgn}(\operatorname{span}\{I,\tilde{A}_x\})$ we do not know whether it is post-hoc self-tested by correlation $\{\braket{\tilde{\psi}|\tilde{A}^{(j)}_x\otimes \tilde{O}^{(k)}|\tilde{\psi}}\}$ based on $\{\tilde{A}_x\}$. Nevertheless, given $\tilde{O}\in \operatorname{sgn}(\operatorname{span}\{\operatorname{sgn}(\operatorname{span}\{I,\tilde{A}_x\})\})$ we can do the following: take the fewest binary observables $\tilde{B}_{|\mathcal{I}_B|},...,\tilde{B}_{Y'-1}\in \operatorname{sgn}(\operatorname{span}\{I,\tilde{A}_x\})$ such that $\operatorname{span}\{I,\tilde{B}_{0},...,\tilde{B}_{Y'-1}\}=\operatorname{span}\{\operatorname{sgn}(\operatorname{span}\{I,\tilde{A}_x\})\}$. Then the correlation Table \ref{table:c} will self-test the corresponding strategy, because the white part tests $\tilde{\mathcal{S}}$, and the green part tests the additional binary observables $\tilde{B}_{|\mathcal{I}_B|},...,\tilde{B}_{Y'-1}\in \operatorname{sgn}(\operatorname{span}\{I,\tilde{A}_x\})$. Now, add $\tilde{O}$ as a new row in the Table \ref{table:d}. Because $\tilde{O}\in \operatorname{sgn}(\operatorname{span}\{\operatorname{sgn}(\operatorname{span}\{I,\tilde{A}_x\})\})$, the yellow part of correlation the Table \ref{table:d} (iteratively) post-hoc self-tests $\tilde{O}$. Thus the correlation Table \ref{table:d} self-tests the extended strategy including $\tilde{O}$. Evidently, via this construction, the size of the correlation table has the trivial upper bound $\frac{d(d+1)}{2}\times \frac{d(d+1)}{2}$ regardless of the number of iterations.

\begin{table}[ht]
\small
\centering
\begin{tabular}{c|c|c|c|c|c|c|c|}
 & $I$ & $\tilde{B}_0$ & ... & $\tilde{B}_{Y-1}$ & \cellcolor{green}{$\tilde{B}_{Y}$} & \cellcolor{green}{...} & \cellcolor{green}{$\tilde{B}_{Y'-1}$} \\
\hline
$I$ & - & $\braket{I, \tilde{B}_0}_{\tilde{\psi}}$ & ... & $\braket{I, \tilde{B}_{Y-1}}_{\tilde{\psi}}$ & \cellcolor{green}{$\braket{I, \tilde{B}_{{Y}}}_{\tilde{\psi}}$} & \cellcolor{green}{...} & \cellcolor{green}{$\braket{I, \tilde{B}_{{Y'-1}}}_{\tilde{\psi}}$}\\
\hline
$\tilde{A}_0$ & $\braket{\tilde{A}_0, I}_{\tilde{\psi}}$ & $\braket{\tilde{A}_0, \tilde{B}_0}_{\tilde{\psi}}$ & ... & $\braket{\tilde{A}_0, \tilde{B}_{Y-1}}_{\tilde{\psi}}$ & \cellcolor{green}{$\braket{\tilde{A}_0, \tilde{B}_{{Y}}}_{\tilde{\psi}}$} & \cellcolor{green}{...} & \cellcolor{green}{$\braket{\tilde{A}_0, \tilde{B}_{{Y'-1}}}_{\tilde{\psi}}$} \\
\hline
... & ... & ... & ... & ... & \cellcolor{green}{...} & \cellcolor{green}{...} & \cellcolor{green}{...} \\
\hline
$\tilde{A}_{X-1}$ & $\braket{\tilde{A}_{X-1}, I}_{\tilde{\psi}}$ & $\braket{\tilde{A}_{X-1}, \tilde{B}_0}_{\tilde{\psi}}$ & ... & $\braket{\tilde{A}_{X-1}, \tilde{B}_{Y-1}}_{\tilde{\psi}}$ & \cellcolor{green}{$\braket{\tilde{A}_{X-1}, \tilde{B}_{{Y}}}_{\tilde{\psi}}$} & \cellcolor{green}{...} & \cellcolor{green}{$\braket{\tilde{A}_{X-1}, \tilde{B}_{{Y'-1}}}_{\tilde{\psi}}$} \\
\hline
 \end{tabular}
\caption{Extended correlation table in the first iteration. We take $X=|\mathcal{I}_A|$, $Y=|\mathcal{I}_B|$, and $Y'=|\mathcal{I}'_B|$ for convenience.}
\label{table:c}
\end{table}

\begin{table}[ht]
\small
\centering
\begin{tabular}{c|c|c|c|c|c|c|c|}
 & $I$ & $\tilde{B}_1$ & ... & $\tilde{B}_{Y-1}$ & \cellcolor{green}{$\tilde{B}_{{Y}}$} & \cellcolor{green}{...} & \cellcolor{green}{$\tilde{B}_{{Y'-1}}$} \\
\hline
$I$ & - & $\braket{I, \tilde{B}_1}_{\tilde{\psi}}$ & ... & $\braket{I, \tilde{B}_{Y-1}}_{\tilde{\psi}}$ & \cellcolor{green}{$\braket{I, \tilde{B}_{{Y}}}_{\tilde{\psi}}$} & \cellcolor{green}{...} & \cellcolor{green}{$\braket{I, \tilde{B}_{{Y'-1}}}_{\tilde{\psi}}$}\\
\hline
$\tilde{A}_1$ & $\braket{\tilde{A}_1, I}_{\tilde{\psi}}$ & $\braket{\tilde{A}_1, \tilde{B}_1}_{\tilde{\psi}}$ & ... & $\braket{\tilde{A}_1, \tilde{B}_{Y-1}}_{\tilde{\psi}}$ & \cellcolor{green}{$\braket{\tilde{A}_1, \tilde{B}_{{Y}}}_{\tilde{\psi}}$} & \cellcolor{green}{...} & \cellcolor{green}{$\braket{\tilde{A}_1, \tilde{B}_{{Y'-1}}}_{\tilde{\psi}}$} \\
\hline
... & ... & ... & ... & ... & \cellcolor{green}{...} & \cellcolor{green}{...} & \cellcolor{green}{...} \\
\hline
$\tilde{A}_{X-1}$ & $\braket{\tilde{A}_{X-1}, I}_{\tilde{\psi}}$ & $\braket{\tilde{A}_{X-1}, \tilde{B}_1}_{\tilde{\psi}}$ & ... & $\braket{\tilde{A}_{X-1}, \tilde{B}_{Y-1}}_{\tilde{\psi}}$ & \cellcolor{green}{$\braket{\tilde{A}_{X-1}, \tilde{B}_{{Y}}}_{\tilde{\psi}}$} & \cellcolor{green}{...} & \cellcolor{green}{$\braket{\tilde{A}_{X-1}, \tilde{B}_{{Y'-1}}}_{\tilde{\psi}}$} \\
\hline
\rowcolor{yellow} $\tilde{O}$ & $\braket{\tilde{O}, I}_{\tilde{\psi}}$ & $\braket{\tilde{O}, \tilde{B}_1}_{\tilde{\psi}}$ & ... & $\braket{\tilde{O}, \tilde{B}_{Y-1}}_{\tilde{\psi}}$ & $\braket{\tilde{O}, \tilde{B}_{{Y}}}_{\tilde{\psi}}$ & ... & $\braket{\tilde{O}, \tilde{B}_{{Y'-1}}}_{\tilde{\psi}}$ \\
\hline
 \end{tabular}
\caption{Extended correlation table in the second iteration.}
\label{table:d}
\end{table} 

\subsection{Self-testing arbitrary real observable}
In \cite{Laura2021Constant}, the authors considered a set of projections summing up to a proportion of $I$, and showed that the strategy consisting of those projections and the maximally entangled state can be self-tested by the correlation it generates. Here we employ one of those strategies with a specific construction. It turns out that, with the initial strategy we chose, in two iterations we will be able to self-test arbitrary binary projective measurement using the iterative scheme.

Consider $d+1$ unit vectors $v_0,\dots,v_{d}\in\mathbb{R}^{d}$ which form the vertices of a regular $(d+1)$-simplex centered at the origin. Note that
\begin{equation}\label{e:simplex}
v_x^\dagger v_{x'}=-\frac1d
\end{equation}
for $x\neq x'$. Define forms $d+1$ binary observables 
$$\tilde{T}_x:=2v_xv_x^\dagger-I.$$ 

The code in Mathematica for generating the observables is provided in {Appendix \ref{sec:appB}}. According to \cite{Laura2021Constant} the following strategy containing $\tilde{T}_x$ and the maximally entangled state is robustly self-tested:
\begin{corollary}
By Theorem 6.10 in \cite{Laura2021Constant}, the strategy $\tilde{\mathcal{S}}^{(0)}=(\ket{\Phi_d},\{\tilde{T}_x\}_{x=0}^{d},\{\tilde{T}_y\}_{y=0}^{d})$ is robustly self-tested by the correlation it generates.
\label{cor:d+1}
\end{corollary}

Now take the strategy $\tilde{\mathcal{S}}^{(0)}$ in Corollary \ref{cor:d+1} as the initial strategy, and consider additional binary observables in the form of $\tilde{T}_{jk}:=\operatorname{sgn}(\tilde{T}_j+\tilde{T}_k)$ for $j\neq k$. By Proposition \ref{prop:conditionbinaryrob} they are robustly post-hoc self-tested. Specifically, we have the following extended strategy that is robust self-tested:
\begin{lemma}
Strategy $\tilde{\mathcal{S}}^{(1)}=(\ket{\tilde{\psi}},\{\tilde{A}_x\}_{x=0}^{d},\{\tilde{B}_y\}_{y=0}^{\frac{d(d+1)}{2}-1})$ is robustly self-tested, where
\begin{align}
    &\ket{\tilde{\psi}}=\ket{\Phi_d},\nonumber\\
    &\{\tilde{A}_x\}_{x=0}^{d}=\{\tilde{T}_x\}_{x=0}^{d},\nonumber\\
    &\{\tilde{B}_y\}_{y=0}^{d}=\{\tilde{T}_y\}_{y=0}^{d},~\{\tilde{B}_y\}_{y=d+1}^{\frac{d(d+1)}{2}-1}=\{\tilde{T}_{jk}:1\le j<k\le d\}\setminus\{\tilde{T}_{12}\}.\nonumber
\end{align}
\label{lem:oneiteration}
\end{lemma}
\begin{proof}
Since $\tilde{\mathcal{S}}^{(0)}=(\ket{\tilde{\psi}},\{\tilde{A}_x\}_{x=0}^{d},\{\tilde{B}_y\}_{y=0}^{d})$ is robustly self-tested, and $\tilde{T}_{jk}\in\operatorname{sgn}(\operatorname{span}\{\tilde{T}_x\})$, by Propositions \ref{prop:conditionbinaryrob} and \ref{prop:extend} we immediately have that the strategy $\tilde{\mathcal{S}}^{(1)}$ is robustly self-tested by the correlation it generates.
\end{proof}

The extended strategy $\tilde{\mathcal{S}}^{(1)}=(\ket{\Phi_d},\{\tilde{A}_x\}_{x=0}^{d},\{\tilde{B}_y\}_{y=0}^{\frac{d(d+1)}{2}-1})$ introduces $d(d+1)/2-d-1$ additional binary observables to Bob that is post-hoc self-tested based on the initial strategy, which are in the form of $\operatorname{sgn}(\tilde{T}_j+\tilde{T}_k)$ (but not every $j\neq k$ is included). It turns out that the additional binary observables together with the $d+1$ initial ones span the space of all $d\times d$ symmetric matrices. To show this, we require the following lemma:

\begin{lemma}
For $d>2$, $\operatorname{span}_{\mathbb{R}}\{\tilde{T}_{jk}:j,k\in[0,d],j\neq k\}=H_d(\mathbb{R})$ the space of all $d$-dimensional symmetric matrices.
\label{lem:spanall}
\end{lemma}

\begin{proof}
Since $\dim H_d(\mathbb{R})=d(d+1)/2=\#\{\tilde{T}_{jk}:j,k\in[0,d],j\neq k\}$, it suffice to show that $\{\operatorname{sgn}(\tilde{T}_j+\tilde{T}_k):0\le j<k\le d\}$ is linearly independent.

Note that $\tilde{T}_j+\tilde{T}_k=2(v_jv_j^\dagger+v_kv_k^\dagger-I)$. Consider the two-dimensional subspace $\mathcal{H}_1=\operatorname{span}(v_j, v_k)$. Then $(\tilde{T}_j+\tilde{T}_k)|_{\mathcal{H}_1^\perp}=-2I$, and 
$$
(\tilde{T}_j+\tilde{T}_k)(v_j-v_k)=\frac{2}{d}(v_j-v_k),\quad (\tilde{T}_j+\tilde{T}_k)(v_j+v_k)=-\frac{2}{d}(v_j+v_k),
$$
so $(\tilde{T}_j+\tilde{T}_k)|_{\mathcal{H}_1}$ has eigenvalues $\pm 2/d$, and the (normalised) eigenvector corresponding to $2/d$ is $w_{jk}:=\sqrt{\frac{d}{2(d+1)}}(v_j-v_k)$. Hence, $\operatorname{sgn}(\tilde{T}_j+\tilde{T}_k)$ have eigenvalues $1$ with multiplicity $1$, and $-1$ with multiplicity $d-1$. Its eigenvector corresponding to $1$ is $w_{jk}=\sqrt{\frac{d}{2(d+1)}}(v_j-v_k)$. Therefore 
$\tilde{T}_{jk}=2w_{jk}w_{jk}^\dagger-I$.

Suppose $\sum_{j<k}c_{jk}\tilde{T}_{jk}=0$ for some real coefficients $c_{jk}$. Then
\begin{equation}
\label{e:lindep}
2\frac{d}{2(d+1)}\sum_{j< k}c_{jk}(v_j-v_k)(v_j^\dagger-v_k^\dagger)=\sum_{j<k}c_{jk}I.
\end{equation}
By Eq. \eqref{e:simplex} we have
$$(v_j^\dagger-v_k^\dagger)v_l =\left\{
\begin{array}{lr}
0 &  j,k\neq l\\
-\tfrac{1}{d}-1     & j\neq k=l\\
1+\tfrac{1}{d}     & l=j\neq k\\
\end{array}
\right.$$
Therefore multiplying Eq. \eqref{e:lindep} by $v_\ell$ on the left results in 
$$\frac{d}{d+1}\left(
\sum_{j<l}c_{jl}(-\tfrac{1}{d}-1)(v_j-v_l)+\sum_{l<k}c_{lk}(1+\tfrac{1}{d})(v_l-v_k)\right)=\sum_{j<k}c_{jk}v_l$$
and so
\begin{align*}
&\left(
\sum_{j<l}c_{jl}(v_l-v_j)+\sum_{l<k}c_{lk}(v_l-v_k)\right)=\sum_{j<k}c_{jk}v_l \\
\Rightarrow\quad&
\left(\sum_{j<l}c_{jl}+\sum_{l<k}c_{lk}-\sum_{j< k}c_{jk}\right)v_l-\sum_{j<l}c_{jl}v_j-\sum_{l<k}c_{lk}v_k=0.
\end{align*}
Since $\sum_j v_j=0$, we further have
$$\left(\sum_{j<l}c_{jl}+\sum_{l<j}c_{lj}-\sum_{j< k}c_{jk}\right)\sum_{j\neq l}v_j-\sum_{j<l}c_{jl}v_j-\sum_{l<j}c_{lj}v_j=0.$$
Since $\{v_j:j\neq l\}$ is linearly independent,
we see that $c_{jl}$ for $j<l$ are equal, and $c_{lj}$ for $l>j$ are equal; therefore $c_{jk}=:c$ for all $j<k$. Thus,
\begin{align*}
    &c\left(d-\frac{d(d+1)}{2}\right)v_l-c\sum_{j\neq l}v_j=0\nonumber\\
    \Rightarrow\quad&c\left(\frac{d-d^2}{2}v_l-\sum_{j\neq l}v_j\right)=0,
\end{align*}
which holds only when $c=0$ or $d=2$ or $d=-1$. So we conclude that $c_{jk}=0$ is the only solution for $\sum_{j<k}c_{jk}\tilde{T}_{jk}=0$ when $d>2$.
\end{proof}

Let $T=\{\tilde{T}_j:j\in[0,d]\}\cup\{\tilde{T}_{jk}:j,k\in[0,d],j\neq k\}$. By Lemma \ref{lem:spanall} we know that $\operatorname{span}_{\mathbb{R}}(T)=H_d(\mathbb{R})$. Note that $|T|=d+1+\frac{d(d+1)}{2}>\dim H_d(\mathbb{R})$. The following proposition gives a maximal linearly independent subset in $T$:
\begin{proposition}
Define $T=\{\tilde{T}_j:j\in[0,d]\}\cup\{\tilde{T}_{jk}:j,k\in[0,d],j\neq k\}$. Let $T'=T\setminus\{\tilde{O}_{0j}:j\in[1,d]\}$ and $T''=T'\setminus\{\tilde{O}_{12}\}$. Then $T''$ is a maximal linearly independent subset in $T$.
\end{proposition}
\begin{proof}
Note that $|T''|=\frac{d(d+1)}{2}=\dim H_d(\mathbb{R})$. So it suffice to show that $\tilde{T}_{0j}\in\operatorname{span}_{\mathbb{R}}(T')$ and $\tilde{T}_{12}\in\operatorname{span}_{\mathbb{R}}(T'')$. Also note that the identity matrix $I=\frac{d}{(d+1)(2-d)}\sum_{j}\tilde{T}_j$ belongs to $\operatorname{span}_{\mathbb{R}}(T'')$, and so does $v_jv_j^*=(\tilde{T}_j+I)/2$.

\begin{itemize}
    \item $\tilde{O}_{0j}\in\operatorname{span}_{\mathbb{R}}(T')$: note that $\sum_jv_j=0$. Then for every $j>0$,
\begin{align*}
    \tilde{O}_{0j}=&\frac{d}{d+1}(v_0-v_j)(v_0^\dagger-v_j)-I\nonumber\\
    =&\frac{d}{d+1}(-\sum_{k>0}v_k-v_j)(-\sum_{k>0}v_k^\dagger-v_j^\dagger)-I\nonumber\\
    =&\frac{d}{d+1}(\sum_{0<k<l}(v_kv_l^\dagger+v_lv_k^\dagger)+\sum_{k>0}(v_kv_j^\dagger+v_jv_k^\dagger)+\tilde{P}_j)-I\in\operatorname{span}_{\mathbb{R}}(T')
\end{align*}
because $v_xv_y^\dagger+v_yv_x^\dagger=\tilde{P}_x+\tilde{P}_y-\frac{d+1}{d}(\tilde{T}_{xy}-I)\in\operatorname{span}_{\mathbb{R}}(T')$ for all $x,y>0$.
    \item $\tilde{T}_{12}\in\operatorname{span}_{\mathbb{R}}(T'')$: we show that $\sum_{1\le j<k\le d}\tilde{T}_{jk}+dv_0v_0^*+\frac{d(d-3)}{2}I=0$, meaning that $\tilde{T}_{12}$ is a linear combination of elements in $T''$.
    Since $\operatorname{span}_{\mathbb{R}}\{v_l:l\in[1,d]\}=\mathbb{R}^d$, it suffices to show that $\sum_{1\le j<k\le d}\tilde{T}_{jk}v_l+dv_0v_0^*v_l+\frac{d(d-3)}{2}v_l=0$ for all $l\in[1,d]$:
\begin{align*}
    &\sum_{1\le j<k\le d}\tilde{T}_{jk}v_l+d\tilde{P}_0v_l+\frac{d(d-3)}{2}v_l\nonumber\\
    =&\sum_{1\le j<k\le d}\frac{d}{d-1}(v_j-v_k)(v_j^\dagger-v_k^\dagger)v_l-\frac{d(d-1)}{2}v_l-v_0+\frac{d(d-3)}{2}v_l\nonumber\\
    =&\sum_{j>0,j\neq l}(v_l-v_j)-\frac{d(d-1)}{2}v_l-v_0+\frac{d(d-3)}{2}v_l\nonumber\\
    =&(d-1)v_l+\sum_{j>0,j\neq l}(-v_j)-\frac{d(d-1)}{2}v_l-v_0+\frac{d(d-3)}{2}v_l\nonumber\\
    =&(d-1)v_l+v_0+v_l-\frac{d(d-1)}{2}v_l-v_0+\frac{d(d-3)}{2}v_l=0.
    \end{align*}
\end{itemize}
\end{proof}

Since $T''=\{\tilde{B}_y\}$ in $\tilde{\mathcal{S}}$ spans the space of all symmetric matrices, every $d$-dimensional binary observable $\tilde{O}_{\operatorname{binary}}$ belongs to $\operatorname{span}\{\tilde{B}_y\}$. Therefore, by adding $\tilde{O}_{\operatorname{binary}}$ into $\{\tilde{A}_x\}$ we construct a strategy that can self-test any binary observable:

\begin{proposition}
For any $d$-dimensional real projective binary measurement, given by observable $\tilde{O}_{\operatorname{binary}}$, the strategy $\tilde{\mathcal{S}}^{(2)}=(\ket{\tilde{\psi}},\{\tilde{A}_x\}_{x=0}^{d+1},\{\tilde{B}_y\}_{x=0}^{\frac{d(d+1)}{2}-1})$ is robustly self-tested, where \begin{align*}
    &\ket{\tilde{\psi}}=\ket{\Phi_d},\nonumber\\
    &\{\tilde{A}_x\}_{x=0}^{d}=\{\tilde{T}_x\}_{x=0}^{d},~\tilde{A}_{d+1}=\tilde{O}_{\operatorname{binary}},\nonumber\\
    &\{\tilde{B}_y\}_{y=0}^{d}=\{\tilde{T}_y\}_{y=0}^{d},~\{\tilde{B}_y\}_{y=d+1}^{\frac{d(d+1)}{2}-1}=\{\tilde{T}_{jk}:1\le j<k\le d\}\setminus\{\tilde{O}_{12}\}.
\end{align*}
\label{prop:arbitrarybinary}
\end{proposition}
\begin{proof}
Since $\tilde{O}_{\operatorname{binary}}\in H_d(\mathbb{R})=\operatorname{span}\{I,\tilde{B}_y\}$ for any $\tilde{O}_{\operatorname{binary}}$, by Lemma \ref{lem:oneiteration} and Proposition \ref{prop:extend} we immediately have that the strategy $\tilde{\mathcal{S}}^{(2)}$ is robustly self-tested by the correlation it generates.
\end{proof}

Finally, we generalize the self-testing of binary measurements to the self-testing of arbitrary $L$-output measurements. Intuitively, we can think of an $L$-output projective measurement as a collection of $L$ binary ones: given an $L$-output projective measurement $\tilde{O}=\sum_{a=0}^{L-1}e^{i2\pi a/L}\tilde{M}_a$, consider binary observables $\{2\tilde{M}_a-I\}_{a=0}^{L-1}$. If we can self-test every binary observable $2\tilde{M}_a-I$, then we should be able to also self-test $\tilde{O}$.

\begin{proposition}
For any $L$-output observable $\tilde{O}=\sum_{a=0}^{L-1}e^{i2\pi a/L}\tilde{M}_a$, strategy $$\tilde{\mathcal{S}}_1=(\ket{\tilde{\psi}},\{\tilde{A}_x,\tilde{O}:x\},\{\tilde{B}_y:y\})$$ 
is robustly self-tested if and only if strategy 
$$\tilde{\mathcal{S}}_2=(\ket{\tilde{\psi}},\{\tilde{A}_x,2\tilde{M}_a-I:x,a\},\{\tilde{B}_y:y\})$$ 
is robustly self-tested.
\label{prop:2toL}
\end{proposition}

\begin{proof}
We prove the `if' part, and the reasoning for the `only if' part is similar.

Suppose $\tilde{\mathcal{S}}_2=(\ket{\tilde{\psi}},\{\tilde{A}_x,2\tilde{M}_a-I:x,a\},\{\tilde{B}_y\})$ is robustly self-tested. Then for any $\varepsilon>0$ there exists $\delta>0$ such that any strategy $\delta/3$-approximately generating correlation $\{\braket{\tilde{\psi}|\tilde{A}_x^{(j)},\tilde{B}_y^{(k)}|\tilde{\psi}}$,\allowbreak$\braket{\tilde{\psi}|(2\tilde{M}_a-I)\otimes\tilde{B}_y^{(k)}|\tilde{\psi}}\}$ can be locally $\varepsilon/L$-dilated by $\tilde{\mathcal{S}}_2$. Let ${\mathcal{S}}_1=(\ket{{\psi}},\{{A}^{(j)}_x,{O}^{(l)}),\{{B}^{(k)}_y\})$ be a strategy that $\delta/3$-approximately generates the correlation $\{\braket{\tilde{\psi}|\tilde{A}_x^{(j)}\otimes\tilde{B}_y^{(k)}|\tilde{\psi}}$, $\braket{\tilde{\psi}|\tilde{O}^l\otimes\tilde{B}_y^{(k)}|\tilde{\psi}}\}$. Construct a physical strategy ${\mathcal{S}}_2:=(\ket{{\psi}},\{{A}_x,2M_a-I),\{{B}_y\})$ where $M_a:=1/L\sum_{l=0}^{L-1}e^{-i2\pi al/L}O^{(l)}$. Then
\begin{align*}
    &|\braket{{\psi}|{O}^{(l)}\otimes{B}_y^{(k)}|{\psi}}-\braket{\tilde{\psi}|\tilde{O}^l\otimes\tilde{B}_y^{(k)}|\tilde{\psi}}|\le\delta/3\nonumber\\
    \Rightarrow\quad&|\braket{{\psi}|(2M_a-I)\otimes\tilde{B}_y^{(k)}|\tilde{\psi}}-\braket{\tilde{\psi}|(2\tilde{M}_a-I)\otimes\tilde{B}_y^{(k)}|\tilde{\psi}}|\le \delta.
\end{align*}
So $\mathcal{S}_2$ $\delta$-approximately generates correlation $\{\braket{\tilde{\psi}|\tilde{A}_x^{(j)}\otimes\tilde{B}_y^{(k)}|\tilde{\psi}},\braket{\tilde{\psi}|(2\tilde{M}_a-I)\otimes\tilde{B}_y^{(k)}|\tilde{\psi}}\}$. By the hypothesis, $\tilde{\mathcal{S}}_2$ is a local $\varepsilon/L$-dilation of $\mathcal{S}_2$, and so
\begin{align*}
    &(\ket{\psi}_{AB},(2{M}_a-I)\otimes I_B)\xhookrightarrow[\varepsilon/L]{\Phi,\ket{\operatorname{aux}}}(\ket{\tilde{\psi}}_{\tilde{A}\tilde{B}},(2\tilde{M}_a-I)\otimes I_{\tilde{B}})\nonumber\\
    \Rightarrow\quad&(\ket{\psi}_{AB},O^{(l)}\otimes I_B)\xhookrightarrow[\varepsilon]{\Phi,\ket{\operatorname{aux}}}(\ket{\tilde{\psi}}_{\tilde{A}\tilde{B}},\tilde{O}^l\otimes I_{\tilde{B}}),
\end{align*}
and 
\begin{align*}
    &(\ket{\psi}_{AB},I_A\otimes B^{(k)}_{y})\xhookrightarrow[\varepsilon/L]{\Phi,\ket{\operatorname{aux}}}(\ket{\tilde{\psi}}_{\tilde{A}\tilde{B}},I_{\tilde{A}}\otimes\tilde{B}^{(k)}_{y}).
\end{align*}
Therefore $\tilde{\mathcal{S}}_1$ is a local $\varepsilon$-dilation of $\mathcal{S}_1$. Thus $\tilde{\mathcal{S}}_1$ is robustly self-tested.
\end{proof}

So we conclude that any real projective measurements can be self-tested:

\begin{theorem}
For any $d$-dimensional $L$-output observable $\tilde{O}$, the strategy $$\tilde{\mathcal{S}}^{(3)}=\left(\ket{\tilde{\psi}},
\{\tilde{A}_x\}_{x=0}^{d+1},
\{\tilde{B}_y\}_{y=0}^{\frac{d(d+1)}{2}-1}\right)$$ is robustly self-tested, where \begin{align*}
    &\ket{\tilde{\psi}}=\ket{\Phi_d},\nonumber\\
    &\{\tilde{A}_x\}_{x=0}^{d}=\{\tilde{T}_x\}_{x=0}^{d},~\tilde{A}_{d+1}=\tilde{O},\nonumber\\
    &\{\tilde{B}_y\}_{y=0}^{d}=\{\tilde{T}_y\}_{x=0}^{d},~\{\tilde{B}_y\}_{y=d+1}^{\frac{d(d+1)}{2}-1}=\{\tilde{T}_{jk}:1\le j<k\le d\}\setminus\{\tilde{T}_{12}\}.
\end{align*}
\label{thm:arbitraryL}
\end{theorem}
\begin{proof}
Statement is true for $L=2$ by Proposition \ref{prop:arbitrarybinary}. For $L>2$, let $\tilde{O}=\sum_ae^{i2\pi a/L}\tilde{M}_a$ be the $L$-output observable, consider the strategy $\tilde{\mathcal{S}}_2=(\ket{\tilde{\psi}},\{\tilde{A}_x\}_{x=0}^{d+L},\{\tilde{B}_y\}_{x=0}^{\frac{d(d+1)}{2}-1})$ where
\begin{align*}
    &\ket{\tilde{\psi}}=\ket{\Phi_d},\nonumber\\
    &\{\tilde{A}_x\}_{x=0}^{d}=\{\tilde{T}_x\}_{x=0}^{d},~\{\tilde{A}_x\}_{x=d+1}^{d+L}=\{2\tilde{M}_a-I:0\le a\le L-1\},\nonumber\\
    &\{\tilde{B}_y\}_{y=0}^{d}=\{\tilde{T}_y\}_{y=0}^{d},~\{\tilde{B}_y\}_{y=d+1}^{\frac{d(d+1)}{2}-1}=\{\tilde{T}_{jk}:1\le j<k\le d\}\setminus\{\tilde{T}_{12}\},
\end{align*}
and $\tilde{M}_a =1/L\sum_{l=0}^{L-1}e^{-i2\pi al/L}\tilde{O}^l$. By a similar argument in the proof of Propisition \ref{prop:arbitrarybinary}, $\tilde{\mathcal{S}}_2$ is robustly self-tested. Then the strategy $\tilde{\mathcal{S}}^{(3)}$ is robustly self-tested by by Proposition \ref{prop:2toL}.
\end{proof}

\section{Iterative self-testing II: general theory}
\label{sec:iterative}

In this section we develop the theory of iterative self-testing in general, whenever the initial state $\ket{\tilde{\psi}}=\ket{\Phi_d}=\sum_j\ket{jj}/\sqrt{d}$ is maximally entangled and all reference measurements are binary and projective, \ie, are described by orthogonal matrices.

Given the initial strategy $\tilde{\mathcal{S}}=(\ket{\Phi_d},\{\tilde{A}_{x}),\{\tilde{B}_{y}\})$, denote $S_0=\{I,\tilde{A}_x\}$ to be the set of initial binary observables, and $V_0=\operatorname{span}_{\mathbb{R}}(S_0)$ to be the subspace generated by $S_0$. Denote by $S'_1=\operatorname{sgn}(V_0)\cap GL_d(\mathbb{R})$ the binary observables that by Proposition \ref{prop:conditionbinaryrob} are post-hoc self-tested on Bob's side, and take $S_1=S'_1\cup\{\tilde{B}_{y}\}$. Note that $S_0\subseteq S'_1\subseteq S_1$ because $\operatorname{sgn}(\tilde{A}_x)=\tilde{A}_x$. 
Let $V_1=\operatorname{span}_{\mathbb{R}}(S_1)$; then we also have $V_0\subseteq V_1$. Now consider the post-hoc self-testing of additional binary observables on Alice's side based on $S_1$; we get the next set of binary observables $S_2=\operatorname{sgn}(V_1)\cap GL_d(\mathbb{R})\supseteq S_1$ that is self-tested, and also the next subspace $V_2=\operatorname{span}_{\mathbb{R}}(S_1)\supseteq V_1$. By iteratively using this technique, we enlarge the set of self-tested binary observables $S_j$ in each step. We remark that, when trying to make a similar argument for a non-maximally entangled $\ket{\tilde{\psi}}$, it is not clear whether $V_j\subseteq V_{j+1}$ still holds.

Since $\{V_j\}_j$ is an increasing sequence of subspaces of the finite-dimensional real Hermitian matrix space $H_d(\mathbb{R})$, it eventually stabilizes at $V_\infty=\lim_{j\rightarrow\infty}V_j$. It is natural to ask, given initial binary observables $\{\tilde{A}_{x}),\{\tilde{B}_{y}\}$, what is $V_\infty$? Before we answer this question, we make the following observation:

\begin{lemma}
Given a set of orthogonal matrices $\{\tilde{A}_{x}\}$, recursively define $V_0=\operatorname{span}_{\mathbb{R}}\{I,\tilde{A}_{x}\}$ and $V_j=\operatorname{span}_{\mathbb{R}}(\operatorname{sgn}(V_{j-1}))$\footnote{
Here we do not exclude the non-singular matrices in $S_j$. In fact, $\operatorname{span}(\operatorname{sgn}(V_j))=\operatorname{span}(\operatorname{sgn}(V_j)\cap GL_d(\mathbb{R}))$: for any singular $s=\operatorname{sgn}(x)$ where $x\in V_j$, $\operatorname{sgn}(x\pm\delta I)\in\operatorname{sgn}(V_j)$. And for small enough $\delta$, we have $\operatorname{sgn}(x\pm\delta I)\in GL_d(\mathbb{R})$ and $2s=\operatorname{sgn}(x+\delta I)+\operatorname{sgn}(x-\delta I)$.
}, where $\operatorname{sgn}$ is defined as in Proposition \ref{prop:conditionbinaryrob}. If $x\in V_j$, then $p(x)\in V_{j+1}$ for any real coefficient polynomial $p\in\mathbb{R}[t]$. 
Consequently, $x,y\in V_j$ implies $xy+yx\in V_{j+1}$.
\label{lem:xinpxin}
\end{lemma}
\begin{proof}
For any $x\in V_j$, let $x=U\Lambda U^\dagger$ where $\Lambda$ has diagonal entries $\lambda_1,\dots,\lambda_d\in\mathbb{R}$ sorted decreasingly. Then $p(x)$ has eigenvalues $p(\lambda_1),\dots,p(\lambda_d)$. Note that the identity matrix $I$ is in $V_1$. Now, for each $i\in[1,d-1]$ such that $\lambda_i\neq\lambda_{i+1}$, choose $r_i\in(\lambda_i,\lambda_{i+1})$, and consider
$$
    x_i:=\operatorname{sgn}(x-r_iI)=U\operatorname{diag}(\underbrace{1,\dots,1}_{i~1s},\underbrace{-1,\dots,-1}_{(d-i)~-1s})U^\dagger.
$$
So $x_i\in \operatorname{sgn}(V_j)\subseteq V_{j+1}$. Since $\{x_i\}$ forms a basis of $\{p(x)\colon p\in\mathbb{R}[t]\}$, we have that $p(x)\in V_{j+1}$ for every $p\in\mathbb{R}[t]$.

Take $p(t)=t^2$, and notice that
$$
xy+yx=\underbrace{(x+y)^2}_{\in V_{j+1}}-\underbrace{x^2}_{\in V_{j+1}}-\underbrace{y^2}_{\in V_{j+1}}{\in V_{j+1}}.
$$
\end{proof}

Lemma \ref{lem:xinpxin} allows one to characterize\footnote{
Here we omit $\{\tilde{B}_{y}\}$ for simplicity. This simplification only strengthens our result because we now do not ask $\{\tilde{B}_{y}\}$ to contribute.} $V_\infty$ in terms of Jordan algebras \cite{jacobson1968structure}. A vector subspace of an associative algebra is a (unital) Jordan algebra if it contains the identity and is closed under the Jordan product $a\star b =\frac12(ab+ba)$.

\begin{proposition}
Given a set of Hermitian orthogonal matrices $\{\tilde{A}_{x}\}$, define $V_0=\operatorname{span}_{\mathbb{R}}\{I,\tilde{A}_{x}\}$, $V_j=\operatorname{span}_{\mathbb{R}}(\operatorname{sgn}(V_{j-1}))$, where $\operatorname{sgn}$ is defined as in Proposition \ref{prop:conditionbinaryrob}. Then $V_\infty=\mathcal{JA}(\{\tilde{A}_x\})$, the real Jordan algebra generated by $\{\tilde{A}_x\}$.
\label{prop:jordan}
\end{proposition}

\begin{proof}
$\mathcal{JA}(\{\tilde{A}_x\})=\mathcal{JA}(\{I,\tilde{A}_x\})$ because $I=\tilde{A}_x^2$. From Lemma \ref{lem:xinpxin} we know that $x,y\in V_\infty$ implies $x\star y\in V_\infty$. So $V_\infty$ is a Jordan algebra, and hence $\mathcal{JA}(\{I,\tilde{A}_x\})\subseteq V_\infty$.

On the other hand, for any $x\in H_d(\mathbb{R})$
the matrix $\operatorname{sgn}(x)$ is a polynomial in $x$, and therefore lies in $\mathcal{JA}(\{x\})$. This implies $V_\infty\subseteq\mathcal{JA}(\{I,\tilde{A}_x\})$. So we conclude that $V_\infty=\mathcal{JA}(\{I,\tilde{A}_x\})=\mathcal{JA}(\{\tilde{A}_x\})$.
\end{proof}

Proposition \ref{prop:jordan} implies that, after sufficiently many steps, every binary observable $\tilde{O}\in\mathcal{JA}(\{I,\tilde{A}_x\})$ can be iteratively post-hoc self-tested based on the binary strategy $\tilde{\mathcal{S}}=(\ket{\Phi_d},\{\tilde{A}_x),\{\tilde{B}_y\})$. 
We also provide two properties of real Jordan algebras that help analysing $\mathcal{JA}(\{\tilde{A}_x\})$.
The first one is that $\mathcal{A}(\{\tilde{A}_x\})$, the real associative algebra generated by $\{\tilde{A}_x\}$, is $M_d(\mathbb{R})$ (the real algebra of $d\times d$ matrices), if and only if $\mathcal{JA}(\{\tilde{A}_x\})$, the real Jordan algebra generated by $\{\tilde{A}_x\}$, is $H_d(\mathbb{R})$ (the real Jordan algebra of symmetric $d\times d$ matrices).

\begin{lemma}\label{lem:fulljordan}
For symmetric $d\times d$ matrices $\{\tilde{A}_x\}$, $\mathcal{A}(\{\tilde{A}_x\})=M_d(\mathbb{R})$ if and only if $\mathcal{JA}(\{\tilde{A}_x\})=H_d(\mathbb{R})$.
\end{lemma}

\begin{proof}
The `if' part: it is straightforward to see that every real matrix is a linear combination of products of symmetric matrices.

The `only if' part:
note that $M_d(\mathbb{R})$ is a simple algebra, and that a Jordan subalgebra of $H_d(\mathbb{R})$ is semisimple. Suppose $\mathcal{A}(\{\tilde{A}_x\})=M_d(\mathbb{R})$. Then we claim that $\mathcal{JA}(\{\tilde{A}_x\})$ is also simple. Indeed; if $\mathcal{JA}(\{\tilde{A}_x\})$ were isomorphic to a non-trivial product of simple ones, then $\mathcal{A}(\{\tilde{A}_x\})$ would likewise be isomorphic to a non-trivial product of simple algebras, which is a contradiction. By the Jordan–von Neumann–Wigner Theorem \cite{Jordan1934Algebraic}, finite-dimensional simple Jordan algebras are isomorphic to one of the following five types:
\begin{itemize}
    \item The Jordan algebra of $n\times n$ Hermitian real matrices $H_n(\mathbb{R})$,
    \item The Jordan algebra of $n\times n$ Hermitian complex matrices $H_n(\mathbb{C})$,
    \item The Jordan algebra of $n\times n$ Hermitian quaternionic matrices $H_n(\mathbb{H})$,
    \item The `spin factor' $\mathbb{R}^n\oplus\mathbb{R}$ with the product $(x,\alpha)\star (y,\beta)=(\beta x+\alpha y,\alpha\beta+\braket{x,y})$,
    \item The Jordan algebra of $3\times 3$ Hermitian octonionic matrices.
\end{itemize}
Since $\mathcal{JA}(\{\tilde{A}_x\})$ is special, we only need to exam the first four cases individually:
\begin{itemize}
    \item $\mathcal{JA}(\{\tilde{A}_x\})\cong H_n(\mathbb{R})$ for some $n$. By the 'if' part of the proof, $\mathcal{A}(\{\tilde{A}_x\})\cong M_n(\mathbb{R})$. But $\mathcal{A}(\{\tilde{A}_x\})=M_d(\mathbb{R})$, so we are only left with the possibility $n=d$.
    \item $\mathcal{JA}(\{\tilde{A}_x\})\cong H_n(\mathbb{C})$ for some $n\ge2$. On one hand, complex Hermitian $n\times n$ matrices do not embed into real matrices of size smaller than $2n$, so $2n\le d$. On the other hand, $\mathcal{JA}(\{\tilde{A}_x\})\cong H_n(\mathbb{C})$ implies that $M_d(\mathbb{R})=\mathcal{A}(\{\tilde{A}_x\})$ is a real subalgebra of $M_n(\mathbb{C})$, so $d^2\le 2n^2$, a contradiction.
     \item $\mathcal{JA}(\{\tilde{A}_x\})\cong H_n(\mathbb{H})$ for some $n\ge2$. On one hand, quaternion Hermitian $n\times n$ matrices do not embed into real matrices of size smaller than $4n$, so $4n\le d$. On the other hand, $\mathcal{JA}(\{\tilde{A}_x\})\cong H_n(\mathbb{H})$ implies that $M_d(\mathbb{R})=\mathcal{A}(\{\tilde{A}_x\})$ is a real subalgebra of $M_n(\mathbb{H})$, so $d^2\le 4n^2$, a contradiction.
    \item $\mathcal{JA}(\{\tilde{A}_x\})$ is a spin factor $\mathbb{R}^n\oplus\mathbb{R}$ for some $n\ge 3$. It is known that it can be embedded in the Hermitian real matrices of size $2^n\times2^n$, but not smaller \cite{mccrimmon2004taste}; therefore $2^n\le d$. 
    On the other hand, the spin factor generates a real Clifford algebra of dimension $2^n$ \cite{jacobson1968structure}, so $2^n\ge d^2$, a contradiction.
\end{itemize}
So, we conclude that $\mathcal{A}(\{\tilde{A}_x\})=M_d(\mathbb{R})$ if and only if $\mathcal{JA}(\{\tilde{A}_x\})=H_d(\mathbb{R})$.
\end{proof}

As a consequence of this lemma, if $\ket{\tilde{\psi}}=\ket{\Phi_d}$, and $\{\tilde{A}_x\}$ generate the real matrix algebra of the corresponding dimension, any binary observable will be in $V_\infty$, thus can be self-tested. In Section \ref{sec:d+1} we showed that a self-tested strategy given by $\cite{Laura2021Constant}$ can be used for this purpose. However, several of the self-tested strategies across the existing literature consist of a maximally entangled state and operators that generate the full matrix algebra, so they can be used to self-test arbitrary observables (of suitable size) by Proposition \ref{prop:jordan}. 
Notice that $\{\tilde{A}_x\}_x$ generates $M_n(\mathbb{R})$ as a real associative algebra if and only if the only real solutions of the linear system [$S\tilde{A}_x=\tilde{A}_xS$ for all $x$] are $S=cI$ the scalar multiples of $I$. Hence given $\{\tilde{A}_x\}$, one can check whether it generates the whole matrix algebra in the following way: suppose we have $X$ binary observables ($d$-dimensional); then  [$S\tilde{A}_x=\tilde{A}_xS$ for all $x$] is a linear system of $d^2$ variables (which are entries of $S$) with $Xd^2$ equations. Thus the condition of Lemma \ref{lem:fulljordan} is equivalent to checking that the coefficient matrix has rank $d^2-1$.

Another property we provide can help in upper-bounding the iteration we need for $V_{itr}=V_\infty$. Let $U_j$ denote the span of all the Jordan products of elements in $S_0$ of length at most $j$. Then we have the following relation between $U_j$ and $V_j$:

\begin{lemma}
For a set of Hermitian orthogonal matrices $S_0=\{I,\tilde{A}_x\}$, define
$$
U_j:=\operatorname{span}_{\mathbb{R}}\{a_1\star\cdots\star a_k,a_i\in S_0,k\le j),
$$
and define $V_j$ as in Proposition \ref{prop:jordan} for $j\ge0$. Then $U_{2^{(j)}}\subseteq V_j$.
\label{lem:jordaniter}
\end{lemma}

\begin{proof}
By definition, $U_1=V_0$. Now suppose $U_{2^{(j)}}\subseteq V_j$ for some $j$. By Proposition \ref{lem:xinpxin} $x\star y\in V_{j+1}$ for every $x,y\in V_j$, so in particular
$x\star y\in V_{j+1}$ for every $x,y\in U_{2^{(j)}}$. 
Since $U_{2^{j+1}}=U_{2\cdot 2^{(j)}}$ is spanned by 
$U_{2^{(j)}}\star U_{2^{(j)}}$,
we conclude that $U_{2^{j+1}}\subseteq V_{j+1}$.
\end{proof}

Note that while $V_j$ is not straightforward to determine (since $\operatorname{sgn}$ is a non-linear map), $U_j$ is easily computable.
If $U_j=U_{j+1}$ for some $j$, then $U_j$ is a Jordan algebra, and so $U_j=V_\infty$; therefore we get an upper bound on the number of iterations as $itr\le\lceil\log_2j\rceil$.
A trivial bound for $U_j$ to stop growing is $\frac{d(d+1)}{2}=\dim H_d(\mathbb{R})$, hence
$$itr\le \left\lceil \log_2\frac{d(d+1)}{2}\right\rceil
\le\lceil 2\log_2 d\rceil.$$

We remark that, for robust self-tested initial strategy with explicit $\varepsilon-\delta$ dependence, we can use Proposition \ref{prop:conditionbinaryrob} repeatedly to get the robustness of the final strategy. For example, some of the robust self-testing results summarized in \cite{Supic2020selftestingof} have robustness $\varepsilon=O(\sqrt{\delta})$, and so $O(C\varepsilon+\delta)=O(\sqrt{\delta})$ in Proposition \ref{prop:conditionbinaryrob}. If we take these initial strategies, we get
$$\varepsilon_\infty=O(\delta^{\frac{1}{2^{itr+1}}})=O(\delta^{\frac{1}{4d}}).$$

Summarizing Proposition \ref{prop:jordan} and Lemma \ref{lem:jordaniter}, and applying Proposition \ref{prop:2toL} to argue about \emph{many-ouput} (rather than just binary-output) measurements, we reach an easy-to-use criterion for a real measurement $\tilde{O}$ to be reachable after iterative self-testing:

\begin{theorem}
Let $\tilde{\mathcal{S}}=(\ket{\Phi_d},\{\tilde{A}_{x}\},\{\tilde{B}_{y}\})$ be a self-tested strategy using maximally entangled state and binary real projective measurements. A real projective measurement $\{\tilde{M}_\ell,\ell\in[0,L-1]\}$ can be iteratively self-tested if 
$$
\tilde{M}_\ell\in\mathcal{JA}(\{\tilde{A}_x\}) \quad  \forall\ell\in[0,L-1],
$$ 
where $\mathcal{JA}(\{\tilde{A}_x\})$ is the real Jordan algebra generated by $\{\tilde{A}_x\}$. Moreover, the number of the iterations is upper-bounded by $\lceil 2\log_2 d\rceil$.
\label{thm:jordan}
\end{theorem}

In particular, if $\mathcal{JA}(\{\tilde{A}_x\})=H_d(\mathbb{R})$, {\it i.e.} $\{\tilde{A}_x\}$ generates the whole real Jordan algebra of symmetric $d\times d$ matrices, then every $d$-dimensional measurement can be self-tested. As what we have shown in Lemma \ref{lem:fulljordan}, it is equivalent to $\{\tilde{A}_x\}$ having a trivial centralizer, which can be checked efficiently.

\section{Conclusion}
\label{sec:conclude}
We have addressed the problem of self-testing an arbitrary real projective measurement by constructing a self-testing protocol using binary measurements and maximally entangled states. Our protocol remains the same for any real projective measurement, provided that the dimension $d$ is fixed. The protocol has a $O(d^3)$-sized question set and a constant-sized answer set. 
We show that our protocol is robust to noise which is crucial from a practical standpoint.
We also introduce a novel theoretical technique of iterative self-testing. This offers a convenient method for establishing new self-tests based on pre-existing ones. Our results show that the set of self-testable observables includes the real Jordan algebra generated by the observables that we utilize for iterative self-testing.

We leave a few open questions and improvements for future work. Now that we know that all real projective measurements can be self-tested, one outstanding challenge is to enhance the efficiency --- specifically, the size and robustness of the protocols. It is known that some high-dimensional states and measurements admit constant-sized self-tests (for example \cite{Laura2021Constant,Fu_2022}, and \cite{sarkar2021self,Supic2021deviceindependent} with constant-sized question sets). Is it the case that all states and measurements can be self-tested by a constant-sized protocol?
Another open question is an explicit $(\delta, \epsilon)$ dependence in the robustness of our self-testing protocol for an arbitrary real projective measurement.
This could have applications for verifiable distributed quantum computation \cite{Reichardt2013}. 

Lastly, from a theoretical standpoint, iterative self-testing is also applicable to strategies with partially entangled states, but the underlying algebraic structure remains to be understood.

\section*{Acknowledgements.}

This work is supported by VILLUM FONDEN via the QMATH Centre of Excellence (Grant No. 10059), Villum Young Investigator grant (No. 37532), 
NSF grant DMS-1954709, and EU Horizon 2020 (Grant No 101017733, VERIqTAS and Grant No 101078107, Qinteract).

\bibliography{ref_for_arxiv}
\bibliographystyle{halpha}

\newpage

\begin{appendices}

\section{Examples and counter examples of post hoc self-testing}
\label{sec:appA}

\subsection{An analytic image of \texorpdfstring{$\operatorname{sgn}$}{sgn} in the two-dimensional case}

Although the image of $\operatorname{sgn}$ is hard to describe in general cases, we give an example where $\operatorname{sgn}(\operatorname{span}\{D^2,DA_x^{(j)}D\})$ has an analytic form. In this case the initial state is a partially entangled, $\ket{\tilde{\psi}}=\cos\gamma\ket{00}+\sin\gamma\ket{11}$ for $\gamma\in (0,\frac{\pi}{4})$, and the binary observable $\tilde{A}_0=X$ the Pauli $X$.

We show that there is a 1-parametric family of post-hoc self-tested binary observables $\operatorname{sgn}(\operatorname{span}\{D^2,DA_x^{(j)}D\})$.  Note that $D\tilde{A}_xD=\sin\gamma\cos\gamma X$ for $\tilde{A_0}=X$. Without loss of generality, suppose $\tilde{O}=\operatorname{sgn}(X+aD^2)$ for some real parameter $a$. If $|a|$ is large, then $X+aD^2$ is diagonally dominant, so $X+aD^2$ will be positive or negative definite, leading to the trivial case $\tilde{O}=\pm I$. So, to obtain a non-trivial $\tilde{O}$, $|a|$ must be bounded, and the upper bound is attained when $X+aD^2$ becomes singular:
\begin{align*}
    \det(X+aD^2)=(a\cos\gamma\sin\gamma)^2-1=0\Rightarrow a=\pm\frac{1}{\cos\gamma\sin\gamma}.
\end{align*}

When $a\in[-\frac{1}{\cos\gamma\sin\gamma},\frac{1}{\cos\gamma\sin\gamma}]$, we can calculate $\tilde{O}$ explicitly as a function of parameter $a$:
\begin{align*}
    \tilde{O}=\begin{bmatrix}
\frac{a\left(-1+2g^{2}\right)}{\sqrt{4+a^{2}\left(1-2 g^{2}\right)^{2}}} & \frac{2}{\sqrt{4+a^{2}\left(1-2 g^{2}\right)^{2}}} \\
\frac{2}{\sqrt{4+a^{2}\left(1-2 g^{2}\right)^{2}}} & \frac{a-2 a g^{2}}{\sqrt{4+a^{2}\left(1-2 g^{2}\right)^{2}}}
\end{bmatrix},
\end{align*}
where $g=\cos\gamma$. Let $\tilde{O}=r_xX+r_zZ$, then $r_x=\frac{2}{\sqrt{4+a^{2}\left(1-2 g^{2}\right)^{2}}}$, ranging from 1 to $\sin2\gamma$. Then in this case, the image of the $\operatorname{sgn}$ is $\{r_xX\pm\sqrt{1-r_x^2}Z\colon \sin(2\gamma)<r_x\le1\}$, which is an uncountable set.

We also give an explicit $\tilde{O}$ that cannot be post-hoc self-tested: let $\gamma=\arctan(1/\sqrt{2})$, $\tilde{O}=H=(X+Z)/\sqrt{2}\not\in \operatorname{sgn}(\operatorname{span}\{D^2,X\})$. Then a ``cheating'' POVM $\{\hat{M}_0,\hat{M}_1\}$ is given by
\begin{align*}
    \hat{M}_0=\begin{bmatrix}
\frac{6-\sqrt{2}}{8} & \frac{\sqrt{2}}{4}\\
\frac{\sqrt{2}}{4} & \frac{\sqrt{2}}{2}
\end{bmatrix},\hat{M}_1=\begin{bmatrix}
\frac{2+\sqrt{2}}{8} & -\frac{\sqrt{2}}{4}\\
-\frac{\sqrt{2}}{4} & 1-\frac{\sqrt{2}}{2}
\end{bmatrix}.
\end{align*}

One can check that this POVM generates the same correlation as $\tilde{O}$, but there is no local isometry connecting them. Indeed; suppose $\Phi[{I}_A\otimes\hat{M}_j\ket{\tilde{\psi}_\gamma}]={I}_A\otimes\tilde{M}_j\ket{\tilde{\psi}_{\gamma}}$ for $j=0,1$, where $\tilde{M}_j=\frac{\tilde{O}+(-1)^jI}{2}$, then we have
$$
0=\braket{\tilde{\psi}_{\gamma}|I_A\otimes\tilde{M}_0\tilde{M}_1|\tilde{\psi}_{\gamma}}=\braket{\tilde{\psi}_{\gamma}|I_A\otimes\hat{M}_0\Phi_A^\dagger\Phi_A\hat{M}_1|\tilde{\psi}_{\gamma}}=\braket{\tilde{\psi}_{\gamma}|I_A\otimes\hat{M}_0\hat{M}_1|\tilde{\psi}_{\gamma}}\neq0,
$$
a contradiction. Thus $\tilde{O}$ cannot be post-hoc self-tested based on $\ket{\tilde{\psi}}$ and $X$. 

\subsection{An obstruction to post-hoc self-testing}

Here we show that, as soon as the number of inputs is small compared to local dimensions, a post-hoc extension of a self-testing strategy is a highly non-trivial phenomenon. In particular, the main theorem is non-trivial, since self-testing does not extend to ``most'' binary observables when the local dimension is large compared to the number of inputs.

\def\R{\mathbb{R}}
\def\C{\mathbb{C}}

\begin{proposition}\label{p:noselftest}
Let $(\ket{\tilde{\psi}},\{\tilde{A}_x\}_{x=0}^{n-1},\{\tilde{B}_y\}_{y=0}^{n-1})$ be a 
$n$-input / 2-output strategy with local dimension $d$. 
Assume that $\ket{\tilde{\psi}}$ has full Schmidt rank and
the observables $\tilde{B}_y$ have a trivial centralizer in $M_d(\R)$.
\begin{enumerate}[(a)]
\item If $\tilde{B}_{-n},\tilde{B}_n\in M_d(\R)$ are distinct binary observables, 
then none of the strategies
$(\ket{\tilde{\psi}},\{ \tilde{A}_x),\{ \tilde{B}_y,\tilde{B}_{\pm n}\})$
is a local dilation of the other.
\item If either $\lfloor \frac{d^2}{4}\rfloor>n+1$
or $\ket{\psi}$ is maximally entangled and $\lfloor \frac{d^2}{4}\rfloor>n$,
then there exist distinct binary observables 
$\tilde{B}_{-n},\tilde{B}_n\in M_d(\R)$
such that the strategies
$(\ket{\tilde{\psi}},\{ \tilde{A}_x),\{ \tilde{B}_y,\tilde{B}_{\pm n}\})$ yield the same correlations,
but none of them is a local dilation of the other (by (a)).
\end{enumerate}
\end{proposition}

\begin{proof}
(a) Suppose $(\ket{\tilde{\psi}},\{\tilde{A}_x),\{\tilde{B}_y,\tilde{B}_{-n}\})$
is a local dilation of
$(\ket{\tilde{\psi}},\{\tilde{A}_x),\{\tilde{B}_y,\tilde{B}_n\})$.
Thus there are isometries $\Phi_A$,$\Phi_B$ and an auxiliary state $\ket{\operatorname{aux}}=\sum_i\sigma_i \ket{ii} \in \R^{d'}$ with $\sigma_1>0$ such that
$\Phi_A\otimes \Phi_B \ket{\tilde{\psi}} = \ket{\operatorname{aux}}\otimes\ket{\tilde{\psi}}$ 
and
$$(\Phi_A\otimes \Phi_B)(\tilde{A}_x\otimes \tilde{B}_y)\ket{\tilde{\psi}} = 
\ket{\operatorname{aux}}\otimes ((\tilde{A}_x\otimes \tilde{B}_{|y|})\ket{\tilde{\psi}}).$$
Then
$$(I\otimes \Phi_B \tilde{B}_y\Phi_B^\dagger)
(\ket{\operatorname{aux}}\otimes\ket{\tilde{\psi}})
=(I\otimes I\otimes \tilde{B}_{|y|})
(\ket{\operatorname{aux}}\otimes\ket{\tilde{\psi}})$$
for all $y\in\{0,\dots,n-1,-n\}$.
Let $\pi:\R^d\otimes \R^{d'}\to \R^d$ be the projection induced by the projection $\R^{d'}\to\R$ onto the first component.
Then
$$(I\otimes \pi\Phi_B \tilde{B}_y\Phi_B^\dagger\pi^\dagger)\ket{\tilde{\psi}}
=(I\otimes \tilde{B}_{|y|})\ket{\tilde{\psi}}$$
Since $\ket{\tilde{\psi}}$ has full Schmidt rank,
$$\pi\Phi_B \tilde{B}_y\Phi_B^{\dagger}\pi^\dagger=\tilde{B}_{|y|}$$
for all $y\in\{0,\dots,n-1,-n\}$.
Let $C=\pi\Phi_B\in\R^{d\times d}$. Note that $C$ is a contraction.
Since $C \tilde{B}_yC^\dagger = \tilde{B}_y$ for all $y\ge0$
and $\tilde{B}_y$ generate the whole $M_d(\R)$ as an $\R$-algebra, it follows that $C$ is invertible (otherwise $\tilde{B}_y$ would have a common kernel). 
Furthermore, if $\lambda$ is an eigenvalue of $C^\dagger$, then $C^\dagger v=\lambda v$ implies
$C (\tilde{B}_yv)=\frac{1}{\lambda} \tilde{B}_yv$ for all $y\ge0$.
Since at least one of $\tilde{B}_yv$ is nonzero if $v\neq0$, it follows that $\frac{1}{\lambda}$ is an eigenvalue of $C$. Since $C,C^\dagger$ are both contractions, we conclude that $C$ is unitary. Therefore
$$C\tilde{B}_y=\tilde{B}_{|y|}C$$
for all $y\in\{0,\dots,n-1,-n\}$.
Since $\tilde{B}_y$ have a trivial centralizer (also in $M_d(\C)$), it follows that $C$ is a scalar multiple of identity. Therefore $\tilde{B}_{-n}=\tilde{B}_n$, a contradiction.

(b) The real algebraic set of binary observables in $M_n(\R)$ has an irreducible component $Z$ of dimension $\lfloor\frac{n^2}{4}\rfloor$ (concretely, $Z$ is the set of binary observables with $\lfloor{\frac{n}{2}\rfloor}$ positive eigenvalues). Consider the map
\begin{equation}\label{e:rag}
Z\to \R^{d+1},\qquad U\mapsto 
(\bra{\psi} \tilde{A}_x\otimes U\ket{\psi},\bra{\psi} I\otimes U\ket{\psi});
\end{equation}
if $\ket{\tilde{\psi}}$ is maximally entangled, one can discard the last component 
$\bra{\psi} I\otimes U\ket{\psi}) = \frac{1}{\sqrt{n}}\tr(U)$
because it is constant on $Z$. Then \eqref{e:rag} is a linear map between semialgebraic sets, so its generic fiber has dimension at least $\lfloor\frac{d^2}{4}\rfloor-n-1>0$ (or $\lfloor\frac{d^2}{4}\rfloor-n>0$ in the maximally entangled case). Therefore there exist distinct $\tilde{B}_{-n},\tilde{B}_n\in Z$ such that 
$$
\bra{\tilde{\psi}} I\otimes \tilde{B}_{-n}\ket{\tilde{\psi}}
=\bra{\tilde{\psi}} I\otimes \tilde{B}_n\ket{\tilde{\psi}},\quad
\bra{\tilde{\psi}}  \tilde{A}_x\otimes \tilde{B}_{-n}\ket{\tilde{\psi}}
=\bra{\tilde{\psi}}  \tilde{A}_x\otimes \tilde{B}_n\ket{\tilde{\psi}}
$$
holds for all $x$.
\end{proof}

If $\ket{ \tilde{\psi}}$ is maximally entangled and $\tilde{A}_y=\tilde{B}_y$, then we know that after sufficiently many post-hoc steps, all binary observables are self-tested (under the given condition on $\tilde{B}_y$). Proposition \ref{p:noselftest}(b) guarantees that this cannot always happen immediately after the first step if number of inputs $n$ is sufficiently smaller than the local dimension $d$; in the case of our preferred strategy with $d+1$ inputs in Section \ref{sec:d+1}, Proposition guarantees ``bad'' binary observables for $d\ge 5$. However, they already exist for $d=3$:

\begin{example}\label{ex:bad3}
Let $\tilde{A}_0,\dots,\tilde{A}_3 \in M_3(\R)$ be the binary observables as in Section 5, and let $\ket{\Phi_3}\in\R^3\otimes\R^3$ be the maximally entangled state (in its Schmidt basis).
Then
$\tilde{S}=
(\ket{\Phi_3},
\{\tilde{A}_x\}_x,\{\tilde{A}_x\}_x\}$
is self-tested by its correlation, and
$\{\tilde{A}_x\}_x$ has trivial centralizer in $M_3(\R)$.
A direct calculation shows that
$$
\tilde{A}_{\pm 4}=\begin{pmatrix}
 0 & -\tfrac{1}{\sqrt{2}} & \pm\tfrac{1}{\sqrt{2}} \\
 -\tfrac{1}{\sqrt{2}} & -\tfrac{1}{2} & \mp\tfrac{1}{2} \\
 \pm\tfrac{1}{\sqrt{2}} & \mp\tfrac{1}{2} & -\tfrac{1}{2} 
\end{pmatrix}$$
are binary observables with one positive eigenvalue, and
$$\bra{\Phi_3}\tilde{A}_x\otimes A_{-4}\ket{\Phi_3}
=\bra{\Phi_3}\tilde{A}_x\otimes A_4\ket{\Phi_3}
\qquad \text{for }x=0,\dots,3.
$$
Therefore the strategies 
$(\ket{\Phi_3},\{ \tilde{A}_x),\{ \tilde{A}_x,A_{\pm 4}\})$
give the same correlation but are not local dilations of each other by Proposition \ref{p:noselftest}(a), so they are not self-tested.
\end{example}

\section{Recipe for the robust self-tested strategy}
\label{sec:appB}
We first show how to construct $d+1$ unit vectors $v_0,\dots,v_{d}\in\mathbb{R}^{d}$ that form the vertices of a regular $d+1$-simplex centered at the origin. This can be guaranteed by $\braket{v_j,v_k}=-1/d$ for all $j\neq k$. To find vectors satisfying this property, consider any unitary $U$ in $\mathbb{R}^{d+1}$ whose first row is the `all one' unit vector $a=(1,1,...,1)/\sqrt{d+1}$. Then, apply $U$ to $d+1$ vectors $\{f_j\}$ where $f_j$ is the normalization of $f'_j=e_j-\braket{a,e_j}a$, and $\{e_j\}_{j=0}^{d}$ are base vectors. We have that all $Uf_j$ are orthogonal to $e_0$. So $\{Uf_j\}_{j=0}^{d}$ spans a $d$-dimensional subspace. We can also show that $\braket{Uf_j,Uf_k}=\braket{f_j,f_k}=-1/d$. So we take $v_x=Uf_x$, $\tilde{P}_x=v_xv_x^\dagger$, and $\tilde{T}_x=2\tilde{P}_x-I$.

\begin{verbatim}
(*local dimension*)
d = 4;
(*find the unitary*)
allone = Normalize[ConstantArray[1, d + 1]];
unitary = ConstantArray[0, {d + 1, d + 1}];
unitary[[1, All]] = allone;
unitary[[2 ;; d + 1, All]] = 
  Table[UnitVector[d + 1, i], {i, 2, d + 1}];
unitary = Orthogonalize[unitary];
(*d+1 vectors*)
vect[x_] := (unitary . 
      Normalize[(UnitVector[d + 1, x] - 
         Projection[UnitVector[d + 1, x], allone])])[[2 ;; d + 1]] // 
   FullSimplify;
(*d+1 projections*)
proj[x_] := Transpose[{vect[x]}] . {vect[x]} // FullSimplify;
(*d+1 binary observables*)
obs[x_] := 2 proj[x] - IdentityMatrix[d];
jordanproduct[x_, y_] := (x . y + y . x)/2;
sgn[x_] := 
  JordanDecomposition[x][[1]] . 
    RealSign[JordanDecomposition[x][[2]]] . 
    Inverse[JordanDecomposition[x][[1]]] // FullSimplify;
(* alternative sgn map *)
(* sgn[x_] := Inverse[x].MatrixPower[x.x,1/2] *)
\end{verbatim}

Based on this one can calculate $\tilde{T}_{jk}=\operatorname{sgn}(\tilde{T}_j+\tilde{T}_k)$, or alternatively 
$\tilde{T}_{jk}=2w_{jk}w_{jk}^\dagger-I$, where $w_{jk}:=\sqrt{\frac{d}{2(d+1)}}(v_j-v_k)$.

\begin{verbatim}
obs2[x_, y_] := sgn[obs[x] + obs[y]];
(*alternative O_{jk} operator*)
(*obs2[x_,y_]:=d/(d+1) \
Transpose[{vect[x]-vect[y]}].{vect[x]-vect[y]}//FullSimplify;*)
\end{verbatim}

\end{appendices}

\end{document}